\documentclass[11pt]{article}
\usepackage{a4wide}
\usepackage{amsmath}
\usepackage{amssymb}
\usepackage{amsthm}
\usepackage{dsfont}
\usepackage{graphicx}
\usepackage{soul}
\usepackage{xcolor}
\usepackage{comment}
\usepackage{algorithm}
\usepackage{algpseudocode}
\usepackage{bbold}
\usepackage{appendix}
\usepackage{todonotes}

\usepackage{hyperref}

\newtheorem{theorem}{Theorem}
\newtheorem{lemma}[theorem]{Lemma}

\newcommand{\OPT}{\mathrm{OPT}}

\newcommand{\polylog}{\text{polylog}(n)}

\renewcommand{\epsilon}{\varepsilon}

\title{Lift-and-Project Integrality Gaps for Santa Claus}

\newcommand{\ifanonymous}[2]{%
#1%
}
\ifanonymous{\author{Etienne Bamas\thanks{Post-Doctoral Fellow, ETH AI Center, Z\"urich, Switzerland.}}}{\author{Anonymous Author(s)}}

\date{}

\begin{document}
\maketitle
\thispagestyle{empty}

\begin{abstract}This paper is devoted to the study of the MaxMinDegree Arborescence (MMDA) problem in layered directed graphs of depth $\ell\le O(\log n/\log \log n)$, which is an important special case of the Santa Claus problem. Obtaining a polylogarithmic approximation for MMDA in polynomial time is of high interest as it is a necessary condition to improve upon the well-known 2-approximation for makespan scheduling on unrelated machines by Lenstra, Shmoys, and Tardos~[FOCS'87].

The only way we have to solve the MMDA problem within a polylogarithmic factor is via an elegant recursive rounding of the $(\ell-1)^{th}$ level of the Sherali-Adams hierarchy, which needs time $n^{O(\ell)}$ to solve. However, it remains plausible that one could obtain a polylogarithmic approximation in polynomial time by using the same rounding with only $1$ round of the Sherali-Adams hierarchy.

As a main result, we rule out this possibility by constructing an MMDA instance of depth $3$ for which an integrality gap of $n^{\Omega(1)}$ survives $1$ round of the Sherali-Adams hierarchy. This result is tight since it is known that after only $2$ rounds the gap is at most polylogarithmic on depth-3 graphs. Second, we show that our instance can be ``lifted'' via a simple trick to MMDA instances of any depth $\ell\in \Omega(1)\cap o(\log n/\log \log n)$ (the whole range of interest), for which we conjecture that an integrality gap of $n^{\Omega(1/\ell)}$ survives $\Omega(\ell)$ rounds of Sherali-Adams. We show a number of intermediate results towards this conjecture, which also suggest that our construction is a significant challenge to the techniques used so far for Santa Claus.

The main inspiration of this work stems from a beautiful construction by Li and Laekhanukit [SODA'22] used in the context of the Directed Steiner Tree problem. Inspired by their construction, we build an MMDA instance of depth $3$ which has interesting properties, and we show how to use the labeling scheme underlying the construction in a novel way to quantify non-trivial correlations between different edges of the graph. Our techniques also seem relevant in the world of Directed Steiner Trees, so we are hopeful they will transfer.

\end{abstract}
\newpage
\section{Introduction}

This paper is devoted to the study of the Santa Claus problem (also known as max-min fair allocation). In this problem, there are gifts (or resources) that need to be assigned
to children (or players). Each gift $j$ has unrelated values $v_{ij}$ for each child $i$. The goal is to assign each gift $j$ to a child $\sigma(j)$ such that we maximize the utility of the least happy child, that is, $\min_i \sum_{j : \sigma(j) = i} v_{ij}$.
The dual of the problem, where one has to minimize the maximum
instead of maximizing the minimum is the problem of makespan
minimization on unrelated parallel machines. Both variants 
form well-known open problems in approximation algorithms~\cite{bansal2017scheduling, schuurman1999polynomial, williamson2011design,Woeginger02}. For the makespan problem, there is a well-known polynomial-time $2$-approximation by Lenstra, Shmoys, and Tardos \cite{lenstra1990approximation}, which has not been improved since, and it is only known that the problem is NP-hard to approximate within a factor better than $3/2$ \cite{lenstra1990approximation}. For the Santa Claus problem, the gap is rather unsatisfactory: polynomial-time algorithms can only guarantee polynomial-factor approximations, and it is only known that the problem is NP-hard to approximate within a factor better than $2$ \cite{bezakova2005allocating,lenstra1990approximation}. Closing the gap between upper and lower-bounds for either one of them is considered an important open question. In fact, it was recently proven in \cite{bamas2024santa} that obtaining a $(2-1/\alpha)$-approximation for the makespan problem implies the existence of an $(\alpha+\epsilon)$-approximation for the Santa Claus problem (for any fixed $\epsilon>0$), at the cost of a polynomial blow-up in the running time. The authors of \cite{bamas2024santa} also show that the converse is true in some significant special case. 

After several attempts at the Santa Claus problem (see e.g. \cite{asadpour2010approximation, golovin2005max,bezakova2005allocating}), the state-of-the-art techniques culminated in the remarkable algorithm by Chakrabarty, Chuzhoy, and Khanna \cite{chakrabarty2009allocating}, which gives an $n^{\epsilon}\cdot \text{polylog}(n)$-approximation in time $n^{O(1/\epsilon)}$, for any $\epsilon=\Omega(1/\log n)$. In particular, for any fixed $\epsilon>0$, this guarantees a $n^{\epsilon}$-approximation in polynomial time, and a $\polylog$-approximation in time $n^{O(\log n/\log \log n)}$. It is not known how to obtain the $\polylog$ guarantee in polynomial time, and the fact is that we do not really understand why.

To give more context, we elaborate on the algorithm in \cite{chakrabarty2009allocating}. It has two main conceptual steps: (i) an intricate reduction to some arborescence problem in layered graphs of depth $O(1/\epsilon)$, and (ii) solving the arborescence problem using lift-and-project hierarchies. Interestingly, step (i) runs in polynomial time and looses a factor $n^{\epsilon}\cdot \polylog$ which is unavoidable, and step (ii) runs in time $n^{O(1/\epsilon)}$, but looses only a $\polylog$ factor.

If we fix $\epsilon= \Theta(\log \log n/\log n)$ in the algorithm, then at the cost of loosing $\polylog$ factors, step (i) reduces general instances of Santa Claus to a slight generalization of the \textit{MaxMinDegree Arborescence problem} (MMDA problem later on). In this problem, we are given a \textit{layered} directed graph $G=(V=L_0\dot \cup L_1\dot \cup \ldots \dot \cup L_\ell,E)$ of size $n$, in which edges can only exist between two consecutive layers $L_i$ and $L_{i+1}$, and oriented from $L_i$ to $L_{i+1}$. There are three types of vertices, one special vertex $s$ called the \textit{source}, some \textit{sinks}, and the rest of the vertices. Further, the depth $\ell$ of the graph has to be at most $ O(\log n/\log \log n)$ (note that it can very well be constant, so the range of interest is essentially $1\le \ell\le O(\log n/\log \log n)$). The goal is then to find an arborescence rooted at the source (i.e. a tree in which all edges are oriented away from the source), such that at each non-sink vertex $u$ selected in the arborescence, the out-degree is at least $k_u/\alpha$, where $\alpha$ is the approximation rate to be minimized and $k_u$ is an integer defined for each vertex $u$. See Figure \ref{fig:MMDA_example} for an example. Finally, we assume that $k_u\ge n^{\Omega(1/\ell)}$ for all $u$, a condition that holds in the instances arising from the reduction of Chakarbarty et al.

\begin{figure}
    \centering
    \includegraphics{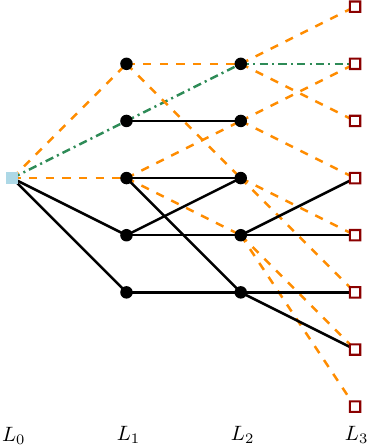}
    \caption{An MMDA instance of depth $3$ with $k_v=2$ for all non-sink vertices. The source is the blue vertex in $L_0$, and the set of sinks is $L_3$. All edges are oriented from left to right. The set of orange edges (dashed) forms an optimum feasible integral solution. The set of green edges (dashed+dotted) is feasible and gives out-degree $1$ to all its non-sink vertices, so it is a $2$-approximate solution.}
    \label{fig:MMDA_example}
\end{figure}

Note that in those instances there is always a trivial polynomial-time $(\max_u k_u)$-approximation by simply selecting a single directed path from the source to a sink, which gives out-degree $1$ to all its vertices. But in our instances, this automatically looses a factor $n^{\Omega(1/\ell)}$ which is $\polylog^{\omega(1)}$ as soon as $\ell=o(\log n/\log \log n)$.

Once we reached this very special case of Santa Claus\footnote{Referring to MMDA as a ``special case'' of Santa Claus is a slight abuse of language, motivated by the reduction above and the fact that there exists an approximation-preserving reduction from MMDA to Santa Claus (see e.g. \cite{bamas2023better,bateni2009maxmin} for details). Formally, it is easy to prove that if there exists an $\alpha$-approximation running in time $T(n)$ for the Santa Claus problem, then there exists an $\alpha$-approximation running in time $T(O(n))+O(n)$ for the MMDA problem; where $n$ refers to the instance size in both problems.}, no further simplification is known. Indeed, the reduction in \cite{chakrabarty2009allocating} would essentially try two things on these layered instances: (i) For all vertices $u$ such that $k_u\le \polylog$, replace $k_u$ by $k'_u=1$ to simplify the instance, and (ii) make $\polylog$ copies of the non-source vertices and arrange them in a layered graph. Here, it is easy to see that (i) does not do anything on our instances, and (ii) does not help either since the graph already has a layered structure (formally the extra copies of each vertex placed in the wrong layer will be unreachable from the source, hence useless). For more details, we refer the reader to the arxiv version of \cite{chakrabarty2009allocating}.

Because of this reduction, the MMDA problem already attracted attention as a prominent special case of the Santa Claus problem, and also as a problem of its own interest (\cite{bateni2009maxmin,bamas2023better}). State-of-the-art algorithms for the problem guarantee a $\polylog$-approximation in time $n^{O(\ell)}$ (see \cite{bateni2009maxmin,chakrabarty2009allocating}), and it is only known that the problem is APX-hard already when $\ell=2$ \cite{bamas2023better}. More recently, the $\polylog$-approximation has even been improved to $\text{polyloglog}(n)$ in the case when $k_u=k$ for all $u$ \cite{bamas2023better}. As explained above, when $k_u\le \polylog$ for all $u$, it is trivial to obtain a $\polylog$ approximation in polynomial time. Otherwise, the most successful algorithms rely on using $\ell-1$ rounds of a certain relaxed version of the Sherali-Adams hierarchy, which we dub the \textit{path hierarchy}. On a high level, these algorithms always proceed layer by layer, starting at the source until reaching the sinks. At layer $i$, for each vertex $v$ which was selected, the algorithm consider the path $p_v$ that was selected from the source to $v$. Then, the path is continued by one more edge using the distribution of edges obtained after conditioning by the event that the path $p_v$ was selected. This justifies the term \textit{path hierarchy}, as only a particular type of conditioning is required: one can write a 
hierarchy of relaxations which only contain a relevant subset of the Sherali-Adams constraints, and which has size equivalent to the number of directed paths of length $\ell$ in the graph. In the worst-case, this is $n^{O(\ell)}$, hence the running time of these algorithms. One can mention that there is also an alternative algorithm (purely combinatorial), \ifanonymous{which is an adaptation of the recursive greedy algorithm in \cite{charikar1999approximation} for Directed Steiner Tree, and guarantees an $O(\ell)$-approximation (\cite{lars_private}) in the case where $k_u=k$ for all $u$, which can be as high as a polylogarithmic approximation.}{which is an easy adaptation of the recursive greedy algorithm in \cite{charikar1999approximation} for Directed Steiner Tree and also gives a polylogarithmic approximation.} However, this algorithm also runs in quasi-polynomial time because it recurses on all possible children of a vertex, and the recursion depth is the depth of the graph $\ell$. This also gives a running time essentially equal to the number of directed paths in the instance.

However, obtaining these guarantees in polynomial time has remained an elusive goal, and it is the main reason why obtaining a $\polylog$-approximation for Santa Claus in polynomial time has been notoriously challenging. There remains a huge gap in our understanding of these techniques. For instance, we do not know how to answer the following basic question:

\begin{center}
    \textit{Is $1$ round of Sherali-Adams enough to solve the MMDA problem within a $O(\polylog)$-factor, regardless of the depth $\ell$?}
\end{center}
Indeed, one could imagine using the round-and-condition algorithm, but only conditioning by the last edge used to reach $v$. It is quite remarkable that the rounding algorithm itself runs in polynomial time, and the super-polynomial running time is only due to the fact that one needs to compute a feasible solution to the hierarchy first. Hence the algorithm which conditions only on one edge runs in polynomial time, and previous works do not rule out the possibility that this could give a poly-logarithmic guarantee. The missing answer to this question is arguably at the heart of our more general misunderstanding of the problem. In fact, we argue later that a positive answer is even plausible if one looks at past works (especially the popular restricted assignment case of Santa Claus). More generally, the issue is that we do not know what a ``difficult'' instance could look like. Hence we also consider a more informal question:
\begin{center}
    \textit{Can the current lift-and-project techniques be used to obtain a $\polylog$-approximation in polynomial-time for the MMDA problem?}
\end{center}

We note that even if one shows that there is an integrality gap surviving after $t=\omega(1)$ rounds of the above hierarchy, this does not necessarily imply that the algorithms above would need to run in time $n^{\Omega(t)}$. Indeed, unlike the general Sherali-Adams hierarchy, the relaxed hierarchy only has a size equivalent to the number of paths of length $t$ in the graph, which could be $n^{O(1)}$ on such an instance. Hence, if one wants to build a meaningful lower bound for these algorithms, the instance must demonstrate an integrality gap that resists many rounds of the relaxed hierarchy, while at the same time having a complex structure to ensure a superpolynomial number of paths. This last condition seems to add some extra difficulty, since all interesting instances exhibited in the literature have a polynomial number of paths (see for instance Section 7 of the arxiv version of \cite{bamas2023better} or some works on the related Directed Steiner Tree problem \cite{halperin2007integrality}).


\subsection{Our results}

We give new constructions which answer the questions above. In all the results about the Sherali-Adams hierarchy, we implicitly refer to the Sherali-Adams hierarchy applied on the naive relaxation of the problem called the \textit{assignment LP} (which will be formally defined in Section \ref{sec:prelim}, along with the path hierarchy). 

\subsubsection{Lower bounds for MMDA via hierarchies}
This part constitutes our main contributions. Here, we focus on MMDA instances in layered graphs. The \textit{depth} of a layered graph is the length of the longest directed path in the graph. We prove the following in Section \ref{sec:sherali1round}.
\begin{theorem}[Main result]
    \label{thm:sherali_1round_gap}
    For any $n$ big enough, there exists a layered graph $G$ of depth $3$ and size $\Theta(n)$ such that $k_u=n^{\Omega(1)}$ for all $u$, and such that the integrality gap of $1$ round of the Sherali-Adams hierarchy is at least $n^{\Omega(1)}$.
\end{theorem}
Note that in the above theorem, we have a lower bound against the general Sherali-Adams hierarchy which is stronger than the path hierarchy. The number of rounds is best possible, since it is known by previous works that $2$ rounds already bring the gap down to $O(\polylog)$ in graphs of depth $3$. The integrality gap is also essentially best possible, since it can never be more than $\max_{u\in V}\{k_u\}$ (remember the trivial algorithm that selects a single directed path).

As a secondary result, we show how the instance of Theorem \ref{thm:sherali_1round_gap} can be ``lifted'' in an elementary way to obtain the following theorem.
\begin{theorem}
    \label{thm:sherali_many_round}
    There exists some absolute constant $c>0$ such that for any $n>c$, and any $\ell\in  \Omega(1)\cap o(\log\log n/\log n)$, there exists a layered graph $G_\ell$ of size $\Theta(n)$ and depth $\ell$ such that the integrality gap of the path hierarchy is still at least $n^{\Omega(1/\ell)}$ after $\ell/c$ rounds.
\end{theorem}
To shed some light on this second result, we note that the path hierarchy already enforces some strong constraints, which makes it non-trivial to obtain this statement. For instance, past works imply that if an edge $e=(u,v)$ remains in the support of the solution at level $t$ of the path hierarchy, then there must exist an integral arborescence rooted at $v$ and satisfying the degree requirements (up to $(\log n)^{O(1)}$ factors) for all vertices at distance at most $t-1$ from $v$ (i.e. there exists a feasible integral solution up to depth $t-1$).

It will be easy to verify that \textit{all} edges remain in the support after $\ell/c$ rounds in the proof of Theorem \ref{thm:sherali_many_round}. This implies that our instance already has the non-trivial property that it contains an integral solution of depth $\Omega(\ell)$ rooted at $v$ for every vertex $v$. It is not obvious to see this from the construction, and it shows that only looking at local parts of the instance is not enough to rule out the existence of an integral solution.

\subsubsection{Relevance of Theorem \ref{thm:sherali_1round_gap}}
In this part, we introduce secondary results which further motivate the significance of our main result, Theorem \ref{thm:sherali_1round_gap}. \ifanonymous{The results in this part are due to Lars Rohwedder \cite{lars_private2}. }{}The proofs and techniques of these are fairly standard and we defer them to Appendix \ref{sec:restricted}. We consider in this part the restricted assignment case of the Santa Claus problem, which is a well-known case where $v_{ij}\in \{0,v_j\}$ for all players $i$ and resources $j$ (see further related works section for references on that special case). In Appendix \ref{sec:restricted}, we show the following.
\begin{theorem}\ifanonymous{\cite{lars_private2}}{}
\label{thm:sherali_restricted_lowerbound}
    For any $n\in  \mathbb N$ and any $\Omega(1/n)\le \epsilon\le 1/3$, there exists an instance of size $n$ of the restricted assignment case such that the integrality gap after $\lfloor\epsilon n\rfloor$ rounds of the Sherali-Adams hierarchy is at least $\Omega(1/\epsilon)$.
\end{theorem}
This implies that a polynomial integrality gap survives a polynomial number of rounds in the Sherali-Adams hierarchy (take for instance $\epsilon=1/\sqrt{n}$ in the above). However, the instances that we use to prove the above theorem do not survive the reduction of Chakrabarty et al., i.e. the reduction will find a way to simplify the instance. Using standard ideas from the restricted assignment, we can even show the following in Appendix \ref{sec:restricted}.

\begin{theorem}\ifanonymous{\cite{lars_private2}}{}
\label{thm:sherali_restricted_upperbound}
    There exists a polynomial-time algorithm which transforms instances of the restricted assignment to instances for which the integrality gap of $1$ round of Sherali-Adams is at most $O(1)$. Furthermore, the transformation looses only a constant-factor in the approximation.
\end{theorem}
The above two theorems show that to obtain a meaningful lower bound for current lift-and-project techniques, it is crucial to focus on instances which cannot be simplified using any known technique. This is precisely why we focus on the MaxMinDegree Arborescence problem for layered graphs. Given Theorem \ref{thm:sherali_restricted_upperbound}, it actually seemed plausible that using only 1 round of Sherali-Adams could be helpful for the general Santa Claus problem. Indeed, the restricted assignment case is already a very challenging case, which has been heavily studied in the past (see further related works). These results give some justification that one round of Sherali-Adams is already quite powerful as it recovers past results on the restricted assignment case. In fact, in Appendix \ref{app:previous_construction} we even give a polynomial gap instance for the configuration LP (the most popular relaxation for Santa Claus, which has an $O(1)$ gap in the restricted assignment case, see e.g. \cite{asadpour2012santa}) which does not survive one round of Sherali-Adams. Together with Theorem \ref{thm:sherali_restricted_upperbound}, this shows that one round of Sherali-Adams is essentially strictly stronger than the configuration LP for the MMDA problem. This explains why our proof of Theorem \ref{thm:sherali_1round_gap} is quite technical. 

\subsubsection{Further consequences of Theorem \ref{thm:sherali_many_round}}

In this part, we discuss additional properties of the instance that we construct to prove Theorem \ref{thm:sherali_many_round}. These properties pose significant challenges to other known techniques for Santa Claus, especially those heavily used in the restricted assignment case. To formally define our result, it is useful to informally introduce a few concepts. A useful trick that appears in previous works on the MMDA problem is to relax the constraint that each vertex can appear at most once in the solution. Instead, one can allow some \textit{congestion} and take a vertex multiple times (i.e. the in-degree of the vertex is more than 1 in the arborescence). With this in mind, we can informally introduce the concept of \textit{locally good solutions}. For any $t>0$, a \textit{$t$-locally good solution} is an integral arborescence which can have congestion as high as $n^{\Theta(1)}$, but such that two different occurrences of the same vertex $v$ appear ``far away'' from each other in the sense that the two directed paths ending at $v$ were disjoint for at least $t$ edges before meeting in $v$ (the formal definition appears in Section \ref{sec:prelim}).

This concept is particularly useful, since it is shown in \cite{bamas2023better} how to ``sparsify'' $\Theta(\ell)$-locally good solutions (where $\ell$ is the depth of the instance) to obtain an integral feasible solution to the MMDA problem which looses only a factor $n^{\Theta(1/\ell)}$ in the approximation rate.

We show the following properties of our construction.
\begin{theorem}
\label{thm:locality_gap}
    Let $G_\ell$ be the MMDA instance of depth $\ell$ used in the proof of Theorem \ref{thm:sherali_many_round}. Then, the following properties hold, where $c$ is an absolute constant:
  \begin{enumerate}
       \item   $k_u=n^{\Theta(1/\ell)}$ for all $u$,
        \item  every vertex $v$ belongs to at least $n^{\Omega(\log (\ell))}$ different directed paths, 
        \item  any feasible integral solution in $G_\ell$ must loose an approximation rate of at least $n^{\Omega(1/\ell)}$, and
        \item  the instance $G_{\ell}$ contains an $(\ell/c)$-locally good solution.
  \end{enumerate}
\end{theorem}

We believe the combination of all these properties to be quite significant. Namely, it seems unlikely that the current lift-and-project techniques combined with the intricate reduction by Chakrabarty et al. \cite{chakrabarty2009allocating} could obtain a $\polylog$-approximation for Santa Claus in time $n^{o(\log \log n)}$. Indeed, by the first property of Theorem \ref{thm:locality_gap}, the reduction technique of Chakrabarty et al. will not modify in any way the instance. Also note that for any $\ell=o(\log n/\log \log n)$, the integrality gap is at least $\polylog^{\omega(1)}$ after $\Omega(\ell)$ rounds of the path hierarchy by Theorem \ref{thm:sherali_many_round}. By the second property of Theorem \ref{thm:locality_gap}, the size of the path hierarchy will still be $n^{\Omega(\log(\ell))}$ after $\Omega(\ell)$ rounds, which is $n^{\Omega(\log \log n)}$ for $\ell$ close to $\log n$. The number of paths also shows that the recursive greedy algorithm \ifanonymous{(\cite{lars_private})}{} does not run in polynomial time on these instances.

Lastly, the third and fourth property of Theorem \ref{thm:locality_gap} show that the analysis of the sparsification method in \cite{bamas2023better} is tight. More generally, these locality-based methods (which are also the intuition behind LP relaxations of the restricted assignment case\footnote{For instance, the well-known configuration LP of \cite{bansal2006santa} essentially strengthens the naive LP by adding the constraints that for any edge $e=(u,v)$ taken in the support, $v$ must have at least $k_v$ outgoing edges in the graph, i.e. there exists a depth-1 solution rooted at $v$. The $t^{th}$ level of the path hierarchy essentially does the same strengthening, but enforces the existence of an integral solution at distance $t$ of $v$ instead of only distance 1.}) also suffer from a gap of $n^{\Omega(1/\ell)}$ on our instances, hence are unlikely to help to obtain a $\polylog$-approximation in polynomial time.


\subsection{Our techniques}
We emphasize the techniques used to prove Theorem \ref{thm:sherali_1round_gap} and Theorem \ref{thm:sherali_many_round} which constitute our main contribution. We note that proving Theorem \ref{thm:sherali_1round_gap} is already quite non-trivial, and that all past MMDA constructions used to obtain an integrality gap for other LP relaxations (see e.g. \cite{bansal2006santa}) do not work (see Appendix \ref{app:previous_construction} for details). The reason for this is simply that the argument that allowed to rule out the existence of an integral solution in those constructions is easily captured by one round of Sherali-Adams.

We highlight here the main ideas of our construction and proof.

\paragraph*{The construction.} An important inspiration to our work is a construction of Li and Laekhanukit \cite{li2022polynomial} who show a polynomial integrality gap of the standard relaxation of the Directed Steiner Tree problem. We take inspiration from their construction to build a layered instance of the MMDA problem of depth $3$ having the properties of Theorem \ref{thm:sherali_1round_gap}. The construction is quite clean, we have $4$ layers of vertices $L_0,L_1,L_2,L_3$, where $L_0$ contains only the source, and the set of sinks is equal to $L_3$. Then, each vertex $v$ in the construction is labeled by a subset $S_v$ of a ground set $\mathcal U=[m]$. The construction is parametrized by a small constant $\rho$ (the reader might think of $\rho =1/1000$ in this overview). There is exactly one vertex in $L_1$ and one in $L_3$ for each subset of $\mathcal U$ of size $\rho m$, and there is exactly one vertex in $L_2$ for each subset of $\mathcal U$ of size $2\rho m$. These labels define the edge set in an easy way, two vertices $u,v$ in consecutive layers are connected if and only if the corresponding labels satisfy that either $S_u\subseteq S_v$, or that $S_v\subseteq S_u$. See Figure \ref{fig:MMDA_gap} for an illustration of the construction. Then we use well-chosen values for the required out-degrees $k_u$, and we show how to use these labels to rule out the existence of a good integral solution.  We note that the authors in \cite{li2022polynomial} did not try to prove integrality gaps for LP hierarchies, so the inspiration from this work stops here. To the best of our knowledge, all the techniques presented beyond this point are completely new.  

\begin{figure}
    \centering
    \includegraphics{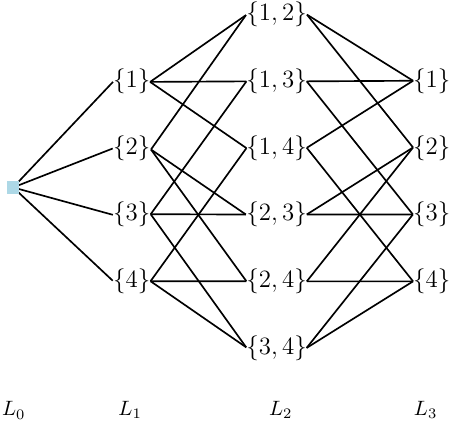}
    \caption{The construction with $\mathcal U=\{1,2,3,4\}$ and $\rho=1/4$.}
    \label{fig:MMDA_gap}
\end{figure}

\paragraph*{Subtree solutions.} Next, we need to go further by using labels to quantify the ``correlations'' between different edges which are needed to obtain a feasible solution to the Sherali-Adams hierarchy. To this end, we introduce the concept of \textit{subtree solutions}: a subtree solution at edge $e=(u,v)$ is a feasible fractional solution $x^{(e)}$ to the naive relaxation of the problem in the same instance, except that we consider the vertex $v$ to be the source vertex (instead of $s$). Intuitively, if $x$ is the fractional solution of the naive relaxation, $x_e$ will be the probability that the edge $e$ appears in the solution, and $x_e^{(e')}$ will be the probability that the edge $e$ appears in the subtree rooted at $e'$ after conditioning by the fact that the edge $e'$ appears. We will also have the convention that $x_e^{(e)}=1$, since we include the edge $e=(u,v)$ in the subtree rooted at $v$. Using this interpretation together with the labels, we are able to quantify very precisely the correlations.

\paragraph*{The shadow distribution.} Then, to show that the instance is feasible for 1 round of the Sherali-Adams hierarchy, we use a standard abstraction of those hierarchies in the form of distributions of edge sets. Intuitively, if we denote by $s_e$ the probability that an edge appears in the set, we are looking for a distribution over edge sets such that $(s_e)_e$ is a feasible solution to the naive LP. Moreover, if we condition by any event of the form ``$e'$ appears in the set'' or ``$e'$ does not appear'' (for any edge $e'$), then we obtain a different vector $(s'_{e})_e$, which also has to be a feasible solution to the naive LP. 

There are two intuitive ways to design a distribution for this purpose. The simplest way is by using the product distribution where each edge $e$ appears independently with probability $x_e$. The second one is to use the distribution given by the round-and-condition algorithm. Both approaches fail, each for different reasons. Informally, the first distribution does not have enough correlation between edges (it does not have any in fact), while the second one has too much. To solve this issue, we design a new distribution, which we call the \textit{shadow distribution}:

\begin{enumerate}
    \item Each edge $e'$ in the graph is selected as a \textit{shadow} edge independently with probability $x_{e'}$. 
    \item Second, for any edge $e'$ which was selected as a shadow edge, we sample a ``subtree''\footnote{The set $S_{e'}$ is sampled using the subtree solution $x^{(e')}$, hence the name. But it is not necessarily a tree.} $S_{e'}$ containing each edge $e$ independently with probability $x_{e}^{(e')}$. 
    \item Then, we return the union of those ``subtrees''.
\end{enumerate}

The intuition behind this distribution is that it constitutes a sweet spot between too much correlation, and too little: If we skip the middle step (step 2), the distribution becomes the product distribution (not enough correlation), and if we repeat step 2 a second time before returning the edge set we essentially obtain the distribution given by the round-and-condition algorithm (too much correlation). The name of ``shadow'' intuitively comes from the fact that the edges in step 1 are hidden and create correlations, but we can only observe the ``subtrees'' which are triggered by these hidden edges.

\paragraph*{Why does it work?} It is non-trivial to see why this works on our instance. However, one key property that we can highlight is that if we denote by $s_e$ the probability that an edge appears in the above distribution, then our instance and subtree solutions are such that 
$x_e\le s_e\le O(x_e)$ for any edge $e$ (i.e. the probability of each edge appearing does not change much compared to the product distribution). Now, it becomes clear that the distribution satisfies the naive LP if we do not condition by any event. Now we can highlight one of the crucial calculation we do in the proof, which relies on our key property above.

The difficulty faced by the independent product distribution is the following: if we fix an edge $e_1\in \delta^-(v)$, and we condition by the event $\mathcal E=\{e_1 \text{ is in the edge set}\}$, then it is necessary that the following property holds:
\begin{equation}
\label{eq:intro}
    \sum_{e\in \delta^+(v)}\mathbb P[\{e \text{ is in the edge set}\}\mid \mathcal E]\ge \Omega(k_v)\ .
\end{equation}
Indeed, if a vertex is selected in the MMDA solution (i.e. has in-degree more than $1$) then it must have out-degree $k_v$ in the solution. It will not be difficult to see in our construction that the above inequality will be violated by a \emph{polynomial} factor in the independent product distribution simply because edges are independent, therefore conditioning by $\mathcal E$ does not increase the probability of other edges showing up.

However, with our shadow distribution, we can reverse the conditioning to obtain the following easy calculation:
\begin{equation*}
    \mathbb P[\{e_1\text{ is a shadow edge}\}\mid \mathcal E]=\frac{ \mathbb P[\mathcal E\mid \{e_1\text{ is a shadow edge}\}]\cdot \mathbb P[e_1\text{ is a shadow edge}]}{\mathbb P[\mathcal E]}\ge \Omega(1)\ ,
\end{equation*}
where the inequality uses the facts that each edge $e$ is a shadow with probability $x_e$, that $x_e^{(e)}=1$ hence $\mathbb P[\mathcal E\mid \{e_1\text{ is a shadow edge}\}]=1$, and finally our key property that $s_e\le O(x_e)$ which imply $\frac{\mathbb P[e_1\text{ is a shadow edge}]}{\mathbb P[\mathcal E]}=\frac{x_{e_1}}{O(x_{e_1})}=\Omega(1)$.

Hence, in the shadow distribution after conditioning by $\mathcal E$, we obtain that with constant probability the edge $e_1$ was a shadow edge. But when $e_1$ is a shadow edge it triggers a subtree rooted at vertex $v$ which has outdegree $k_v$ in expectation, which will satisfy Equation \eqref{eq:intro}.

We believe that the property that $x_e\le s_e\le O(x_e)$ is something very special about this construction. For instance, one can show that the fact that there exists a subtree solution $x^{(e')}$ for every $e'$ is not enough to guarantee this property, and also not enough to fool 1 round of the Sherali-Adams hierarchy (see Appendix \ref{app:subtree} for an easy example). 

\paragraph*{The lifted instance.} For Theorem \ref{thm:sherali_many_round}, we need to construct an instance with many layers, as it is known that only $2$ rounds of the path hierarchy is already strong enough to not be fooled by our construction of depth $3$. For this, we use the labels of the previous instance as inspiration. We note that the instance has intuitively $2$ phases. An \textit{expanding} phase, where the labels correspond to bigger and bigger subsets from size $0$ to $2\rho m$, and a \textit{collapsing} phase where the size of the labels decreases from $2\rho m$ to $\rho m$. Then, we use a simple trick to refine the set system that gives the labels. Essentially, we fix some small $\epsilon\ge \Omega(1/m)$ (note that $m=\Theta(\log n)$ where $n$ is the instance size). Then, we build an instance with $\Theta(1/\epsilon)$ layers, with one expanding and one collapsing phase as before, where the $i$-th layer of the expanding phase contains a vertex for each set of size $i\epsilon m$, and the $i$-th layer of the collapsing phase contains a vertex for each set of size $2\rho m-i\epsilon m$. This allows us to obtain in some sense a ``continuous'' version of our instance of depth $3$. Interestingly, a lot of properties of the instance of depth $3$ carry over to this new instance. For instance, the proof to rule out integral solutions in basically the same. Then, one can define the same concept of subtree solutions, and a generalization of our probability distribution for $t$ rounds of Sherali-Adams is essentially as follows: 
\begin{enumerate}
    \item Sample a shadow set $S_1$ containing each edge $e$ independently with probability $x_e$.
    \item For $i=2$ to $t$, sample a shadow set $S_{i}$ as follows: each edge $e'\in S_{i-1}$ creates a set $S_{i}^{(e')}$ which contains each edge independently with probability $x_{e}^{(e')}$, and we set $$S_i=\bigcup_{e'\in S_{i-1}}S_{i-1}^{(e')}\ .$$
    \item Return $S_t$.
\end{enumerate}

We conjecture that this fools the $t$-th round of the Sherali-Adams hierarchy on our instances of depth $\Omega(t)$, but have not been able to prove it. The main issue is that computing the probabilities of events that several edges appear becomes very intricate.

However, we were able to use the intuition that this distribution gives to fool the path hierarchy, which seems to allow more approximations. We proceed in a similar way using the concept of subtree solution. The path hierarchy only has a variable for each set of edges which form a directed path $p=e_1,e_2,\ldots, e_k$ (with $k\le t$). The question is, what is the probability of the event that $p\subseteq S_t$? It appears difficult to compute, but we believe that this probability is well-approximated (within $\polylog$ factors) by the probability of the event 
\begin{equation*}
    \mathcal E = \{e_1\in S_1\}\cap \{e_2\in S_{2}^{(e_1)}\}\cap \ldots \cap \{e_k\in S_{k}^{(e_{k-1})}\}\ .
\end{equation*}
It is easy to see that 
\begin{equation*}
  \mathbb P[\mathcal E] = x_{e_1}\cdot \prod_{i=2}^{k}x_{e_i}^{(e_{i-1})}\ .
\end{equation*}
We show that setting the variable $y(p)$ exactly equal to the above probability yields a feasible solution to the path hierarchy.

\subsection{A closely related problem: the Directed Steiner Tree problem}

Another well-known and not really understood problem is the Directed Steiner Tree problem, in which one has to find the cheapest directed tree connecting a root to all terminals. We elaborate more on the state-of-the-art results for this problem, as the parallel with Santa Claus is quite striking. The best algorithms give a $n^{\epsilon}\cdot \text{polylog}(n)$-approximation in time $n^{O(1/\epsilon)}$, for any $\epsilon=\Omega(1/\log n)$ \cite{charikar1999approximation,grandoni2019log2,rothvoss2011directed,friggstad2014linear}. The main difference with Santa Claus is that the DST problem cannot be approximated within a factor better than $\log^{2-\epsilon}(n)$ for any fixed $\epsilon>0$ unless $\textrm{NP}\subseteq \textrm{DTIME}(n^{\polylog})$ \cite{halperin2003polylogarithmic}. What is even more remarkable is the similarity of the techniques used to solve the problem. A well-known ``height reduction theorem'' (\cite{zelikovsky1997series,calinescu2005polymatroid}) reduces general instances of the problem to instances on layered graphs of depth at most $\ell$ by loosing a factor $O(\ell\cdot n^{1/\ell})$. For $\ell=O(\log n)$, this looses a factor $O(\log n)$. Then, the state-of-the-art algorithm by \cite{grandoni2019log2} essentially proceeds layer by layer from the root to the terminals, rounding and conditioning each time on a proper subset variables. Specifically, when reaching vertex $v$, the algorithm select as the conditioning set to be the path which was selected from the source to $v$. Hence, the concept of path hierarchy also makes sense in that context. Moreover, we note that it is possible to adapt the recursive greedy algorithm of \cite{charikar1999approximation} to the MMDA problem \ifanonymous{(\cite{lars_private})}{}. Lastly, as explained our results were inspired by a very clean construction of \cite{li2022polynomial} (itself taking inspiration from \cite{zosin2002directed}). Our trick of refining the set system can be seen as a lifting trick for our $3$-layered instance, and in principle it can also be applied in the context of DST.

\subsection{Further related works}
As mentioned in introduction, and important special case of Santa Claus is the restricted assignment where each gift $j$ has a fixed value $v_j$, and each child $i$ can only access a subset $\mathcal R(i)$ of the gifts. Equivalently, this is the case where $v_{ij}\in\{0,v_j\}$ for all $i,j$. This case is fairly well-understood, with a long line of work (see e.g. \cite{davies2020tale,bansal2006santa,feige2008allocations,annamalai2017combinatorial,DBLP:conf/icalp/ChengM18,DBLP:conf/icalp/ChengM19,haxell2022improved,Polacek}) culminating in a $(4+\epsilon)$-approximation in polynomial time. There are also works on the restricted assignment case with non-linear utility functions, such as a $O(\log \log n)$-approximation in polynomial time for the case of submodular utilities \cite{bamas2021submodular}. Those techniques were then transferred to the makespan scheduling problem to obtain a better-than-$2$ approximation in quasi-polynomial time, and a better-than-$2$ estimation algorithm in polynomial time (see e.g. \cite{svensson2012santa,JansenR19,annamalai2019lazy,bamas2024santa,chakrabarty20141}), again in the restricted assignment case.

As other related works, one can cite the Densest-$k$-Subgraph problem which admits a polynomial-approximation in polynomial time \cite{bhaskara2010detecting} and share some common points with our arborescence problem. Some lift-and-project lower bounds were already showed \cite{bhaskara2012polynomial}, however the constructions and proofs are very different to ours. Finally, one can mention some routing problems (e.g. orienteering with time windows) for which the state-of-the-art techniques are reminiscent of those for the MMDA problem (see e.g. \cite{chekuri2005recursive}). 

\subsection{Overview of the paper}
In Section \ref{sec:prelim}, we give some formal definitions of all the concepts that were used in introduction, and more. In Section \ref{sec:sherali1round} we prove Theorem \ref{thm:sherali_1round_gap}. In Section \ref{sec:construction}, we give our general construction for any depth $\ell$, and we prove a few properties of the obtained instances. In Section \ref{sec:lower_bound_many_rounds}, we prove Theorems \ref{thm:sherali_many_round} and \ref{thm:locality_gap}. In Appendix \ref{sec:restricted} we prove Theorems \ref{thm:sherali_restricted_lowerbound} and \ref{thm:sherali_restricted_upperbound}.

To get better intuition of the constructions, we would advise the reader to read Section \ref{sec:sherali1round} before reading Sections \ref{sec:lower_bound_many_rounds} and \ref{sec:construction}.

\section{Preliminaries}
\label{sec:prelim}
\subsection{Problem definition}
In the rest of this paper (with the exception of Appendix \ref{sec:restricted}), we will work on the MaxMinDegree Arborescence (MMDA) problem on a layered directed graph $G=(V,E)$ of size $n$ and of depth $\ell\le O(\log n/\log \log n)$. By \textit{layered}, we mean that the graph $G$ contains some \textit{layers} of vertices $L_0\cup L_1\cup \ldots \cup L_\ell=V$, and edges can only exist between two consecutive layers $L_i$ and $L_{i+1}$. Further, the edge has to be oriented from $L_i$ to $L_{i+1}$. In our constructions, the sinks will all belong to the last layer $L_\ell$, and the layer $L_{0}$ contains only the source vertex $s$. Hence, the set of sinks is equal to $L_{\ell}$. Further, our constructions will have the property that $k_u\ge n^{\Omega(1/\ell)}$ for all $u\in V\setminus L_{\ell}$.

One can describe the input/output for the MMDA problem as follows.

\paragraph*{Input.} The input of the problem is the graph $G$, the description of every vertex as either a sink, a source, or a normal vertex, and the required out-degrees $k_v$. 

\paragraph*{Output.} An $\alpha$-approximate solution to the problem is a subgraph $T\subseteq E$ such that:
\begin{enumerate}
    \item The source $s$ has out-degree (in $T$) $|\delta^+_T(s)|\ge \alpha k_s$,
    \item every vertex $v\in V$ has in-degree (in $T$) at most $1$, and
    \item for every vertex $v$ with $|\delta^-_T(v)|=1$, we have $|\delta^+_T(v)|\ge \alpha k_v$.
\end{enumerate}
One seeks to find the biggest $\alpha^*$ possible so that there exists an $\alpha^*$-approximate solution. 

\subsection{Basic notations}

We will say that a vertex $u$ is \textit{reachable} from $v$ in the graph $G$ if there exists a directed path from $v$ to $u$. We will denote by $A(v)$ the set of \textit{ancestors} of vertex $v$ which contains all the vertices from which $v$ is reachable. Similarly, $D(v)$ is the set of \textit{descendants} of $v$ which are all the vertices reachable from $v$.

By a slight abuse of notation, we will also denote by $A(v)$ the set of edges $e=(u,u')$ such that $u'\in A(v)$, and by $D(v)$ the set of edges $e=(u,u')$ such that $u\in D(v)$. We will denote by $D(e)$ the set of vertices $x$ such that $x\in D(u')$. We will also denote by $D(e)$ the set of edges $e'=(x,y)$ such that $x\in D(u')$. It will always be clear whether the considered object is an edge or a vertex, and this information is sufficient to clear up the ambiguity.

We also define the intuitive notations $L_{\ge i}=\cup_{j\ge i}L_j$, $L_{\le i}=\cup_{j\le i}L_j$, $L_{< i}=V\setminus \cup_{j\ge i}L_j$, $L_{> i}=V\setminus \cup_{j\le i}L_j$. For an edge $e=(u,v)$, we say that $e\in L_i$ if $v\in L_i$ (i.e. if the endpoint of $e$ belongs to layer $L_i$). By a similar overloading of the notation, we have the same definitions of $L_{\ge i}, L_{\le i}, L_{> i}, L_{< i}$ for edges.

For any vertex $v$, we use the standard notation of $\delta^+(v)$ for the edges going out of $v$, $\delta^-(v)$ for the edges going in, and $\delta(v)=\delta^+(v)\cup \delta^-(v)$.

Finally, in many of our construction, each vertex $v$ will be labeled by a set (over a ground set of elements), which we will denote by $S_v$. Similarly, for any edge $e=(u,v)$, we denote by $S_e$ the set labeling $v$, i.e. $S_e:=S_v$.

\subsection{Relaxations of the problem}

\paragraph{Naive relaxation.} The naive relaxation of the arborescence problem is called the \textit{assignment LP}, and reads as follows on our layered instances (recall that $L_\ell$ is the set of sinks).

\begin{align*}
    x(\delta^+(s)) &\ge k_s \\
    x(\delta^+(v)) &\ge  k_v\cdot x(\delta^-(v)) &\forall v\in V\setminus (\{s\}\cup L_{\ell})\\
    x(\delta^-(v)) &\le 1 &\forall v\in V\\
    0\le x_e&\le 1 &\forall e\in E
\end{align*}

It is not difficult to see that this relaxation has a polynomial integrality gap already if $\ell=2$. We call the first two rows the \textit{covering constraints}, and the third one the \textit{packing constraints}.

We note that in several proofs, we might obtain solutions that satisfy the constraints up to a multiplicative factor $\alpha\ge 1$, we say that the solution is \textit{$\alpha$-approximate}. The factor by which the packing constraint is violated will often be called the \textit{congestion}. By standard arguments, this can be easily transformed into a feasible solution to the assignment LP with value $k'_u\ge k_u/(\alpha)^2$ at every vertex. To see this, first we scale down all $k_u$s by some factor $\alpha$. The new covering constraints are now satisfied. Let us now scale down a second time the demand at the source $k'_s$ and all fractional values $x_e$ by some factor $\alpha$. Now, it is clear to see that all the constraints are satisfied, and that $k'_u\ge k_u/(\alpha)^2$ for all $u\in V$.

A similar argument holds for \textit{integral} solutions of the assignment LP. If we have an integral solution which is $\alpha$-approximate w.r.t. the assignment LP, then by standard flow arguments, one can transform it in polynomial time into a feasible integral solution where every vertex $u$ in the solution receives out-degree at least $\lfloor k_u/(\alpha)^2\rfloor$ (see e.g. \cite{chakrabarty2009allocating}). In all our solutions, we will always have that $\alpha=o(k_u)$, so the floor function will not have any significant impact on the approximation factor. Given these remarks, we will only verify the constraints up to some multiplicative factor in the rest of the paper.

\paragraph{Sherali-Adams lift-and-project.} For clarity, we will not write down the result of the Sherali-Adams hierarchy on the assignment LP. However, we state and prove here the result that we use to prove the feasibility of one round of Sherali-Adams.
\begin{theorem}
\label{thm:sherali_distributions}
Consider the relaxation of a binary integer linear program.
Then there is a feasible solution for $r$
levels of SA if
there exists a probability distribution $D$ over $0/1$ assignments of the variables
such that for all variable sets $|V_0\dot\cup V_1|\le r$ 
and every constraint $a^T x \le b$ of the linear program we have that
    \begin{equation*}
      \mathbb E_{x\sim D}[a^T x \mid E] \le b \ ,
    \end{equation*}
    where $E$ be the event that $x_i = 0$ for all $i\in V_0$ and $x_i = 1$ for all $i\in V_1$
    and $\mathbb P_{x\sim D}[E] > 0$.
\end{theorem}
Abstractions of Sherali-Adams in the form of distributions
are very standard. Similar theorems are proven e.g.\ in~\cite{georgiou2008limitations}.
\begin{proof}
    For all variable sets $V_0 \dot\cup V_1$ define 
    $s_{V_0, V_1} = \mathbb P_{x \sim D}[x_i = 0\forall i\in V_0 \text{ and } x_i = 1\forall i\in V_1]$.
    Then $s_{V_1} := s_{\emptyset, V_1}$ for $|V_1|\le r+1$ will be our lifted SA variables and in particular $s_i := s_{\{i\}}$ are the variables of the original LP that survive $r$ rounds of SA.
    
    We claim that for all constraints $a^T x \le b$, all $|V_0\dot\cup V_1| \le r$, and
    $z = \prod_{i\in V_1} x_i \prod_{i\in V_0} (1 - x_i)$,
    we have
    \begin{equation*}
        z * a^T s = \sum_{i} a_i \cdot s_{V_0, V_1\cup\{i\}} \text{ and } z * b = s_{V_0, V_1} \cdot b \ ,
    \end{equation*}
    where $*$ is the SA multiplication operator.
    We argue inductively over $|V_0|$. For $V_0 = \emptyset$ it is obviously true.
    For $|V_0| > 1$ let $z = (1 - x_j) z'$ for some $j\in V_0$. Then
    \begin{align*}
        z * a^T s &= z' * a^T s - x_j * z' * a^T s \\
        &= \sum_{i} a_i \cdot s_{V_0\setminus\{j\}, V_1\cup \{i\}}
        - \sum_{i} a_i \cdot s_{V_0 \setminus \{j\}, V_1\cup \{i, j\}} \\
        &= \sum_{i} a_i (s_{V_0\setminus\{j\}, V_1\cup \{i\}} - s_{V_0 \setminus \{j\}, V_1\cup \{i, j\}})
        = \sum_{i} a_i \cdot s_{V_0, V_1\cup\{i\}}
    \end{align*}
    Here, the second equation comes from the induction hypothesis and the last
    equation follows from the definition of $s_{U, W}$.
    Similarly,
    \begin{equation*}
        z * b = z' * b - x_j * z' * b = b (s_{V_0\setminus\{j\}, V_1} - s_{V_0\setminus\{j\}, V_1\cup\{j\}}) = b \cdot s_{V_0, V_1}
    \end{equation*}
    We will now verify the lifted constraint
    $z * a^T x \le z * b$ and some
    $z = \prod_{i\in V_1} x_i \prod_{i\in V_0} (1 - x_i)$,
    where $|V_0\dot\cup V_1| \le r$.
    We have that
    \begin{align*}
        z * a^T s &= \sum_{i} a_i \cdot s_{V_0, V_1\cup\{i\}} \\
        &= \sum_{i} a_i \cdot \mathbb P_{x \sim D}[x_j = 1\forall j\in V_1\cup\{i\} \text{ and } x_j = 0 \forall j\in V_0] \\
        &= \sum_{x: x_j = 1\forall j\in V_1 \text{ and } x_j = 0 \forall j\in V_0} \mathbb P_{x \sim D}[x] \cdot \sum_{i: x_i = 1} a_i \\
        &= \mathbb P_{x \sim D}[x_j = 1\forall j\in V_1 \text{ and } x_j = 0 \forall j\in V_0] \cdot \mathbb E_{x \sim D}[a^T x \mid x_j = 1\forall j\in V_1 \text{ and } x_j = 0 \forall j\in V_0] \\
        &= s_{V_0,V_1} \cdot b \\
        &= z * b \ . 
    \end{align*}
\end{proof}
Note that Theorem \ref{thm:sherali_distributions} has stronger requirements than usually needed. In general, one does not require the existence of a single distribution for all constraint and conditioning, but rather the existence of one distribution for each constraint and conditioning, with the added condition that different distributions have to be consistent with each other.

\paragraph{The path hierarchy.} Here we define a \textit{path hierarchy}, which is can be seen a weakened variant of the Sherali-Adams hierarchy and has been used in previous works \cite{bamas2023better,bateni2009maxmin,chakrabarty2009allocating} (see \cite{bateni2009maxmin} specifically to see why the Sherali-Adams hierarchy imply the path hierarchy). 

We need a few definitions first. For a directed path $p$, we denote by $C(p)$ the set of children paths of $p$, which is the set of paths of length $|p|+1$ and which contain $p$ as a prefix (i.e. continue $p$ by one more edge). We denote by $D(p)$ the set of \textit{descendant paths} of $p$, which are simply the paths which contain $p$ as a prefix. For any vertex $v$, $I(v)$ is the set of paths which end at $v$. We denote by $P$ the set of directed paths in the instance. For any path $p$ ending at some vertex $v$, we write $k_p:=k_v$ to be the degree requirement at the endpoint of $p$. 

Now we can state the path LP hierarchy. For clarity, we assume that there is a dummy edge $e_0$ incoming at the source $s$. We have a variable $y(p)$ for each directed path $p$.
\begin{align}
    \sum_{q\in C(p)} y(q) &= k_p \cdot y(p) &\forall p\in P \label{eqLP:flowLPdemand}\\
     \sum_{q\in I(v)\cap D(p)} y(q) &\leq y(p) &\forall p\in P,v\in V \label{eqLP:flowLPcapacity}\\
    \sum_{e\in \delta^-(v)}y(\{e\}) &\le 1 & \forall v\in V \label{eqLP:flowLPpacking_unlifted}\\
     \sum_{e\in \delta^+(v)}y(\{e\}) &\ge k_v\cdot \sum_{e\in \delta^-(v)}y(\{e\}) & \forall v\in V \label{eqLP:flowLPcovering_unlifted}\\
    y(q)&\le y(p) & \forall p,q\in P,\  p\subseteq q \label{eqLP:consistency}\\
    y(\{e_0\}) &= 1 & \label{eqLP:flowLPdemandroot}\\
    0\le y&\le 1 
\end{align}

To understand those constraints, we think of the variables $y(p)$ as binary variables which indicate whether the path $p$ is selected in the solution. Constraint \eqref{eqLP:flowLPdemand} is simply a covering constraint which implies that if a path $p=(...,v)$ is selected, there must be $k_v$ paths continuing $p$ by one more edge. Constraint \eqref{eqLP:flowLPcapacity} is a packing constraint which states that if a path $p$ is selected, no vertex can have more than one incoming path inside the subtree rooted at $p$. Constraint \eqref{eqLP:flowLPpacking_unlifted} and Constraint \eqref{eqLP:flowLPcovering_unlifted} are simply the standard assignment LP constraint. Constraint \eqref{eqLP:consistency} are consistency constraints important to be able to build consistent distributions. Indeed, if a path $p$ is not selected in the integral solution, a path $q$ containing $p$ cannot be selected either. Finally, Constraint \eqref{eqLP:flowLPdemandroot} ensures that the source is covered.

The $t^{th}$ level of the path hierarchy is the relaxation above where we removed all constraints which contain variables for paths of length strictly more than $t+1$. It is known by previous works that $\ell-1$ levels suffice to obtain a $\polylog$-approximation on graphs of depth $\ell$ \cite{bateni2009maxmin,chakrabarty2009allocating,bamas2023better}. In fact, previous works show that the $t^{th}$ level of the path hierarchy enforces the following remarkable constraint: for any edge $e=(u,v)$ such that $y(\{e\})>0$, there must exists an integral arborescence rooted at $v$ of depth at least $t$ which has congestion $O(\polylog)$. This is in stark constrast with the configuration LP \cite{bansal2006santa} which only enforces such constraints at depth $1$.

\subsection{Locally good solutions}

It will be useful in some parts to think of a feasible solution as a set of directed paths as follows: a solution $T$ is transformed into a set of directed paths $P_T$ which contains all directed paths in $T$ starting at the source $s$. For some $t\ge 1$, we say that a solution $T$ is $t$-\textit{locally good} if the set of path $P_T$ obtained from $T$ satisfies the following constraints:
\begin{enumerate}
    \item For any $p\in P_T$ where $p$ ends at a vertex $v$, there are at least $k_v$ paths of length $|p|+1$ containing $p$ as a prefix, and
    \item For any path $p\in P_T$ ending at vertex $u$, and any vertex $v$ at distance at most $t$ from $u$, we have that the number of paths $q\in P_T$ containing $p$ as a prefix and ending at $v$ is at most $O(\polylog)$. 
\end{enumerate}
\cite{bamas2023better} show that such a $t$-locally good solution can be transformed into a feasible integral solution by loosing an approximation rate of at most $n^{\Theta(1/t)}$.

\subsection{Subtree solutions}

An important concept which appears in many of our proofs is the notion of \textit{subtree solution.} A subtree solution $x^{(v)}$ for some vertex $v$ is a feasible solution to the assignment LP, except that we move the source at vertex $v$. Hence, for any edge $e$, we will build subtree solutions such that $x_e^{(v)}>0$ if and only if $e\in D(v)$. This is very intuitive indeed, one can see that it does not help the solution to have $x_e^{(v)}>0$ if there is not a single path from $v$ to $e$.

One can see that the standard assignment LP solution $x$ is a subtree solution for the source $s$ (i.e. we can, and will, set $x^{(s)}=x$).

By a slight abuse of notation, for any edge $e=(u,v)$ we will denote by $x^{(e)}$ the subtree solution $x^{(v)}$. One can intuitively think of this subtree solution as setting $x_e^{(e)}=1$, and then extending it with the $x^{(v)}$ values. This is equivalent to setting the source at vertex $v$, and having one dummy edge with fractional value $1$ ending at vertex $v$.

\subsection{A few technical lemmas}

In several of our constructions, we will use a certain labeling scheme, where each label correspond to some set of elements of a ground set $[m]$. Our instances will have size $2^{\Theta(m)}$. The following lemmas will be useful to bound certain quantities.

In the rest of the paper, for any $0\le x\le 1$, we define the entropy function of a Bernoulli random variable of parameter $x$ as $h(x):=-x\log_2(x)-(1-x)\log_2(1-x)$, where $\log_2$ is the logarithm in base $2$. The following lemmas can be obtained from standard techniques, and some version of them already appear in \cite{li2022polynomial}.

\begin{lemma}[Stirling's approximation]
    For any integer $n>0$, $$\log_2(n!)=n\log_2(n)-n\log_2(e)+O(\log n)\ .$$
\end{lemma}

\begin{lemma} \label{lem:binomial_coeff}
For any $\alpha,\beta \in (0,1)$ with $\alpha\ge \beta$,
    \begin{align*}
    \log_2 {\alpha m \choose \beta m} = (\alpha m)\cdot h\left( \beta/\alpha\right)\pm O(\log m)\ .
\end{align*}
\end{lemma}
\begin{proof}
    We simply use Stirling's formula.
\end{proof}

Lemma \ref{lem:binomial_coeff} allows us for a slight abuse of notation. For the ease of exposition in the rest of this paper, we might have some binomial coefficients ${p_m\choose q_m}$ without checking that $p_m,q_m$ are integers. In all our the rest of the paper, one can think that we take the size of the ground set $m$ big enough that all quantities are integers, or that when we refer to this binomial coefficient, we can replace it by the value $2^{p_m\times h(q_m/p_m)}$.

\begin{lemma} \label{lem:binomial_function}
Consider some fixed constants $\alpha,\beta,\gamma \in (0,1)$. Let $f$ be the function
\begin{equation*}
    f:j\mapsto  {\beta m \choose j}{\alpha m \choose \gamma m-j}
\end{equation*} over the integers $j\in [\max\{0,(\gamma-\alpha)m\},\min\{\gamma,\beta\} m]$. Denote by $M$ the maximum of the function over its domain. Then, we have
\begin{enumerate}
    \item $$M=m^{O(1)}\cdot {\beta m\choose \frac{\gamma \beta}{\alpha+\beta}m}{\alpha m\choose \frac{\gamma \alpha}{\alpha+\beta}m}\ ,$$
    \item For any $j$ in the domain and any $\delta=O(1)$, we have that 
    \begin{equation*}
        \frac{f(j)}{f(j\pm \delta)}\le m^{\Theta(1)}\ , \text{ and}
    \end{equation*}
    \item There exists some $\delta=O(1)$ such that the function $f$ is increasing on the interval $[\max\{0,(\gamma-\alpha)m\},\frac{\gamma \beta}{\alpha+\beta}m-\delta]$, and decreasing on the interval $[\frac{\gamma \beta}{\alpha+\beta}m+\delta,\min\{\gamma,\beta\} m]$.
\end{enumerate}
\end{lemma}
\begin{proof}
    Let us prove the second statement first. Let us consider some non-negative $\delta=\Theta(1)$. We compute 
    \begin{align*}
        \frac{f(j)}{f(j+\delta)}&= \frac{{\beta m \choose j}{\alpha m \choose \gamma m-j}}{{\beta m \choose j+\delta}{\alpha m \choose \gamma m-j-\delta}}\\
        &=\frac{(j+\delta)!}{j!}\cdot \frac{(\beta m - j-\delta)!}{(\beta m-j)!}\cdot \frac{(\gamma m - j -\delta)!}{(\gamma m - j)!}\cdot \frac{(\alpha m - \gamma m +j + \delta )!}{(\alpha m - \gamma m +j)!}\\
        &\le \frac{(j+\delta)!}{j!}\cdot \frac{(\alpha m - \gamma m +j + \delta )!}{(\alpha m - \gamma m +j)!}\\
        &\le (j+\delta)^{\delta}\cdot (\alpha m +j + \delta )^{\delta}\\
        &\le m^{O(1)}.
    \end{align*}
    If $\delta=\Theta(1)$ is negative, we obtain in the same way 
    \begin{align*}
        \frac{f(j)}{f(j+\delta)}&\le  \frac{(\beta m - j-\delta)!}{(\beta m-j)!}\cdot \frac{(\gamma m - j -\delta)!}{(\gamma m - j)!}\\
        &\le (\beta m-\delta)^{-\delta}\cdot (\gamma m - \delta )^{-\delta}\\
        &\le m^{O(1)}.
    \end{align*}
    For the third statement, we compute
    \begin{align*}
        &\frac{f(j+1)}{f(j)} = \frac{(\beta m - j)(\gamma m - j)}{(j+1)(\alpha m -\gamma m + j +1)}\ge 1\\
        &\iff (\beta m - j)(\gamma m - j)\ge (j+1)(\alpha m -\gamma m + j +1)\\
        &\iff j\le \frac{\beta \gamma m^2-\alpha m+\gamma m - 1}{(\alpha + \beta)m+2} = \frac{\beta \gamma m^2 (1+O(1/m))}{(\alpha + \beta)m (1+O(1/m))}=  \frac{\beta \gamma m }{\alpha + \beta}+\Theta(1)\ .
    \end{align*}
    For the first property, we use the other two properties which were proven. By the third property, the maximum is around the point $j^*:=\frac{\beta \gamma m }{\alpha + \beta}$, and we only loose a factor $m^{O(1)}$ by using this approximation using the second property.
\end{proof}

\section{A lower bound for one round of Sherali-Adams}
\label{sec:sherali1round}

In this section, we describe a lower bound for $1$ round of the Sherali-Adams hierarchy on the naive assignment LP. We will use our construction from Section \ref{sec:construction} with $\ell=3$. For clarity, we describe here completely the instance. We only prove in this section that the constructed instance contains a feasible solution to the Sherali-Adams hierarchy (up to some $(\log n)^{O(1)}$ factor). We prove in Section \ref{sec:construction} (specifically Lemma \ref{lem:no_integral_sol}) that there is no integral solution with an approximation rate $n^{o(1)}$.

\subsection{The instance}
$\rho$ will be chosen as a small enough constant. We have a ground set $\mathcal U=[m]$ of size $m$ where $\rho m\in\mathbb N$. Our graph has $4$ layers $L_0,L_1,L_2,L_3$ of vertices. The layer $L_0$ contains only one vertex, the source $s$. $L_3$ is exactly the set of all sinks in this instance.

\paragraph{Vertex set.} In $L_0$, there is only the source $s$. In $L_1$, there is exactly one vertex $u$ for each set of $\rho m$ elements of the ground set $\mathcal U$ (hence ${m\choose \rho m}$ vertices). In $L_2$, there is exactly one vertex $u$ for each set of $2\rho m$ elements of the ground set $\mathcal U$ (hence ${m\choose 2\rho m}$ vertices). In $L_3$, there is exactly one \textit{sink} $u$ for each set of $\rho m$ elements of the ground set $\mathcal U$ (hence ${m\choose \rho m}$ sinks). For any vertex, we denote by $S_u$ the set associated to $u$ (one can think that $S_s=\emptyset$). 

\paragraph{Edge set.} There is a directed edge from $s$ to every vertex in $L_1$. There is a directed edge $(u,v)$ from $u\in L_1$ to $v\in L_2$ if and only if $S_u\subseteq S_v$. Similarly, there is a directed edge $(u,v)$ from $u\in L_2$ to $v\in L_3$ if and only if $S_v\subseteq S_u$.

\paragraph{Required out-degrees.} The source requires 
\begin{equation*}
    k_s:=\frac{{m\choose \rho m}}{{(1-\rho)m\choose \rho m}}
\end{equation*} outgoing edges to be covered in the arborescence. Each vertex $u\in L_1$ needs 
\begin{equation*}
    k_u=k_1:=\frac{{(1-\rho)m\choose \rho m}}{{2\rho m\choose \rho m}}
\end{equation*} outgoing edges to be covered. Finally, each vertex $u\in L_2$ needs out-degree
\begin{equation*}
    k_u=k_2:={2\rho m\choose \rho m}
\end{equation*} to be covered.

Clearly, the total size of the instance is $n=2^{\Theta (m)}$, and we have that $k_u=n^{\Omega(1)}$ for all $u$.

\subsection{Feasibility of the assignment LP}
For convenience, we restate here the naive assignment LP, before defining our subtree solutions. The assignment LP reads as follows.
 \begin{align*}
    x(\delta^+(s)) &\ge k_s \\
    x(\delta^+(v)) &\ge k_v\cdot x(\delta^-(v)) &\forall v\in V\setminus (\{s\}\cup T)\\
    x(\delta^-(v)) &\le 1 &\forall v\in V\\
    0\le x_e&\le 1 &\forall e\in E\ .
\end{align*}

Recall that $L_i$ is the set of edges whose endpoint is in $L_i$. It is easy to verify that $x$ defined by
\begin{equation*}
    x_e := \begin{cases}
        {(1-\rho)m \choose \rho m}^{-1} &\text{ if } e\in L_1, \\
        {(1-\rho)m \choose \rho m}^{-1} \cdot {2\rho m \choose \rho m}^{-1} &\text{ if } e\in L_2, \\
        {(1-\rho) m \choose \rho m}^{-1} &\text{ if } e\in L_3
        \end{cases}
\end{equation*}
is a feasible solution to this LP, with $k_s= {m \choose \rho m}{(1-\rho)m \choose \rho m}^{-1}$, $k_v={(1-\rho)m \choose \rho m}{2\rho m\choose \rho m}^{-1}$ for all $v\in L_1$, $k_v={2\rho m\choose \rho m}$ for all $v\in L_2$.

\subsection{Feasibility of SA(1)}

In this section, we prove that there exists a feasible solution to the first level of the Sherali-Adams hierarchy. We show this by constructing a distribution over possible edge subsets $A$ which survives certain conditioning together with Theorem \ref{thm:sherali_distributions}.

Before proceeding, it is useful to consider two natural distributions which fail, but are useful for intuition.

\subsubsection{Failed attempt 1}

The first idea that comes to mind is simply the independent distribution: each edge is in $A$ independently with probability $x_e$.

Clearly, this distribution satisfy all constraints in expectation if we do not condition by any event. However, there is a fundamental issue if we condition by an event of the form $\mathcal E=\{e_1\in A\}$ for some $e_1=(s,v)\in L_1$ for instance. Indeed, if we consider the covering constraint at vertex $v$, we need that 
\begin{equation*}
    \sum_{e\in \delta^+(v)} \mathbb P[e\in A \mid \mathcal E]\ge k_1 \cdot \mathbb P[e_1\in A \mid \mathcal E]=k_1\ .
\end{equation*}
But since the edges appear in $A$ independently of each other, we obtain that 
\begin{equation*}
    \sum_{e\in \delta^+(v)} \mathbb P[e\in A \mid \mathcal E] =  \sum_{e\in \delta^+(v)} \mathbb P[e\in A]= \sum_{e\in \delta^+(v)} x_e = k_v\cdot x_{e_1}\ .
\end{equation*}
Now remember that $x_{e_1}=n^{-\Omega(1)}$, so this distribution looses a polynomial factor. It highlights that some correlations are necessary to fool the Sherali-Adams hierarchy.

\subsubsection{Failed attempt 2}

Given the rounding algorithms in previous works, a natural idea to fix the issue and introduce correlations would be to sample the edges according to what the rounding algorithm would do. More precisely, we would proceed as follows.
\begin{enumerate}
    \item Sample each edge $e=(s,v)\in L_1$ independently with probability $x_e$.
    \item For each vertex $v\in L_1$ which has an incoming edge from the previous step, sample each edge $e\in \delta^+(v)$ independently with probability ${2\rho m\choose \rho m}^{-1}$.
    \item For each vertex $v\in L_2$ which has an incoming edge from the previous step, select all the edges $e\in \delta^+(v)$.
\end{enumerate}

The sampling probabilities at each step are selected so that even after conditioning by some event $\mathcal E=\{e_1\in A\}$, the covering constraints remain satisfied. This distribution will fix the issue that independent distributions had, but another problem appears. Indeed, let us consider the event $\mathcal E=\{e_1\in A\}$ for some $e_1=(s,v)\in L_1$, and the packing constraints at the sink $t$ such that $S_t=S_v$. In that case, we need to verify that

\begin{equation*}
    \sum_{e\in \delta^-(t)} \mathbb P[e\in A \mid \mathcal E]\le 1\ .
\end{equation*}

We claim that, $\mathbb P[e\in A \mid \mathcal E]\ge {2\rho m\choose \rho m}^{-1}$ for all $e\in \delta^-(t)$. Indeed, since $S_t=S_v$, for all edges $e=(u,t)\in \delta^-(t)$, we have that $(v,u)\in E$. Hence we write
\begin{equation*}
    \mathbb P[e\in A \mid \mathcal E] \ge  \mathbb P[e\in A \mid \{(v,u)\in A\}]\cdot  \mathbb P[\{(v,u)\in A\} \mid \mathcal E] = 1\cdot {2\rho m\choose \rho m}^{-1}\ .
\end{equation*}

Finally, we note that $\delta^-(t)={(1-\rho )m\choose \rho m}$, therefore 
\begin{equation*}
    \sum_{e\in \delta^-(t)} \mathbb P[e\in A \mid \mathcal E]\ge {(1-\rho )m\choose \rho m} \cdot {2\rho m\choose \rho m}^{-1} = n^{\Omega(1)}\ ,
\end{equation*}
and we loose a polynomial factor again. The intuition is that the rounding algorithm creates too much correlations as it has ``2 layers'' of correlation. Intuitively, we need a distribution which has only one layer of correlation. Before introducing our distribution, we define our subtree solutions, which will be crucial to our distribution.

\subsubsection{Subtree solutions}
For any edge $e=(u,v)$, we define $S_e=S_v$ the set corresponding to the endpoint of $e$. We also recall that $D(e)$ is the set of descendant edges of $e$ (i.e. the set of edges $e'$ such that there exists a directed path starting at $e$ and ending at $e'$).

\paragraph{Subtree solution for $e=(u,v)\in L_1$.} For any edge $e\in L_1$, we set 
\begin{equation*}
    x_{e'}^{(e)}=\begin{cases}
        1 \mbox{ if } e'=e\\
        {2\rho m\choose \rho m}^{-1} \mbox{ if } e'\in L_2\cap D(e)\\
        \frac{{(1-\rho )m\choose \rho m}}{{m\choose \rho m}}\cdot {m-2\rho m+|S_{e'}\cap S_e|\choose |S_{e'}\cap S_e|}^{-1}\mbox{ if } e'\in L_3\cap D(e)\\
        0 \mbox{ otherwise.}
    \end{cases}
\end{equation*}

We claim that the above subtree solution is a feasible solution to the assignment LP of the instance given by the same graph, where the source placed at the endpoint $v$ of the edge $e=(u,v)$. Most of the constraints are easy to check, except the packing constraint $ x(\delta^-(v)) \le 1$ for $v\in L_3$, and the covering constraint  $x(\delta^+(v)) = k_v\cdot x(\delta^-(v))$ for $v\in L_2$. 

To see why these constraints are feasible, imagine the following process. For all edges in $\delta^+(v)$, we set the fractional value to ${2\rho m\choose \rho m}^{-1}$. Therefore, every vertex $u'$ in $L_2$ such that $(v,u')\in E$ has in-degree ${2\rho m\choose \rho m}^{-1}$ and needs fractional out-degree equal to $1$ (because $k_{u'}={2\rho m\choose \rho m}$). Then, all the sinks push one unit of flow uniformly towards the vertices in $L_2\cap D(v)$ (recall that $D(v)$ are the vertices reachable from $v$). Formally, if some sink $t$ has $d$ neighbors in $L_2\cap D(v)$, we set $x_{e'}^{(e)}=1/d$ for all edges $e'$ going from one of those neighbors. One can compute this number, as it is exactly equal to the number of sets of size $2\rho m$ containing both $S_t$ and $S_e$, i.e. this is equal to ${m-2\rho m+|S_t\cap S_e|\choose |S_t\cap S_e|}$, hence the expression. Clearly, no vertex has fractional in-degree more than 1, hence the packing constraint is satisfied. We need one last crucial check, which is to verify that the vertices in $L_2\cap D(v)$ indeed have fractional out-degree at least $1$. To see this, note the following two observations:
\begin{enumerate}
    \item By symmetry, all vertices in $L_2\cap D(v)$ receive the same out-degree. This is because all vertices in $L_2\cap D(v)$ are entirely defined by a set of size $2\rho m$ containing $S_e$, which have no special property, except that they all contain $S_e$. Formally, for any vertex $u\in L_2\cap D(v)$, one can count that there are exactly ${\rho m\choose j}{\rho m\choose \rho m -j}$ sinks $t$ connected to $u$ such that $|S_t\cap S_v|=j$. This number does not depend on $u$.
    \item All the sinks in $L_3$ are reachable from $v$. This is because a sink $t$ is reachable from $v$ if, and only if there exists a set of size $2\rho m$ containing both $S_v$ and $S_t$. But $|S_v\cup S_t|\le |S_v|+|S_t|=2\rho m$, so there always exists such a set. 
\end{enumerate}
Combining the above two observations, we conclude that each vertex in $L_2\cap D(v)$ receives a fractional out-degree exactly equal to 
\begin{equation*}
    \frac{|T|}{|L_2\cap D(v)|}=\frac{{m\choose \rho m}}{{(1-\rho)m\choose \rho m}}\ ,
\end{equation*}
hence we can even afford to scale down by the factor $ \frac{{(1-\rho )m\choose \rho m}}{{m\choose \rho m}}$ and still get a feasible solution. This is exactly how we set the variables for the edges in $L_3$.

\paragraph{Subtree solution for $e=(u,v)\in L_2$.} If $e\in L_2$, we set 
\begin{equation*}
    x_{e'}^{(e)}=\begin{cases}
        1 \mbox{ if } e'=e\\
        1 \mbox{ if } e'\in L_3\cap D(e)\\
        0 \mbox{ otherwise.}
    \end{cases}
\end{equation*}
It is easy to check all the constraints.

\paragraph{Subtree solution for $e=(u,v)\in L_3$.} If $e\in L_3$, we set 
\begin{equation*}
    x_{e'}^{(e)}=\begin{cases}
        1 \mbox{ if } e'=e\\
        0 \mbox{ otherwise.}
    \end{cases}
\end{equation*}
The constraints are trivially satisfied.

\subsubsection{The shadow distribution}

Using subtree solutions, we are ready to define our distribution of edges to fool the Sherali-Adams hierarchy via Theorem \ref{thm:sherali_distributions}. We sample a set $A$ of \textit{active} edges in the following way.

\begin{enumerate}
    \item We select a set of \textit{shadow} edges $S$ which contains each edge $e\in E$ independently with probability $x_e$. 

    \item For each shadow edge $e\in S$, we create a set $S_e$ which contains each edge $e'\in E$ independently with probability $x_{e'}^{(e)}$ (recall that $x_e^{(e)}=1$, so that $e\in S_e$ with probability 1). We will say that $e\in S$ ``triggers a subtree''. 

    \item We return the set $A=\cup_{e\in S}S_e$.
\end{enumerate}

In the following, it will also be useful to define $A'$ to be the multiset of active edges (indeed, an edge might be selected several times in the above process). 

The rest of this section is to prove that the shadow distribution with our choice of subtree solutions works. We start by a few useful lemmas.

\subsubsection{Some useful lemmas}
In this part, we prove some useful lemmas which will help later. From now on, we denote by $A'$ the set of active edges counted with their multiplicity (indeed, note that an edge can appear in the subtree of several shadow edges).

We denote by $\mathcal E$ the event by which Sherali-Adams can condition. Note that it must be of the form $\mathcal E=\{e_1\in A\}$ or $\mathcal E=\{e_1\notin A\}$ for some edge $e_1$. In the former case, we say that $\mathcal E$ is a \textit{positive} event, and a \textit{negative} event in the latter case.

We define by $s_e$ the probability of the event that $\{e\in A\}$ in the above probability distribution. We start with the following crucial lemma. Note that if this lemma was not true, it would not even be clear how to verify the constraints without conditioning by any event $\mathcal E$. 
\begin{lemma}
    \label{lem:sherali_lifteddistribution1}
    For any edge $e\in E$, we have that $x_e\le s_e\le 6x_e$.
\end{lemma}
\begin{proof}
First, we note that $s_e\ge \mathbb P\left[\{e\in S\}\cap \{e\in S_e\}\right]=x_e$. For the upper-bound, we compute the probability $s_e$ of any edge $e$ appearing in the set $A$. We clearly have that 
    \begin{align*}
        \mathbb P[e\in A]&=1-\mathbb P[e\notin A]\\
        &=1-\mathbb P\left[\bigcap_{e'\in E} \{e'\notin S\}\cup \{e\notin S_{e'}\}\right]\\
        &= 1 -\prod_{e'\in E} \mathbb P\left[\{e'\notin S\}\cup \{e\notin S_{e'}\}\right]\\
        &= 1-\prod_{e'\in E} \left(1-\mathbb P\left[\{e'\in S\}\cap \{e\in S_{e'}\}\right]\right)\\
        &= 1-\prod_{e'\in E} \left(1-x_{e'}\cdot x_e^{(e')}\right)\ ,
    \end{align*}
    where the third equality uses the fact that each edge $e'$ is selected as a shadow independently of other edges, and that each edge $e'$ selects its set $S_{e'}$ independently of other edges. Hence, we have 
\begin{align*}
    s_e = 1-\prod_{e'\in E}\left(1-x_{e'}\cdot x_e^{(e')}\right)\ .
\end{align*}
If $e\in L_1$, then note that $e$ is given a non-zero value in a subtree solution $x^{(e')}$ if and only if $e'=e$. Therefore, we have
\begin{equation*}
    s_e = 1- \left(1-x_e\cdot x_e^{(e)}\right) = 1- \left(1-x_e\right)=x_e\ .
\end{equation*}
If $e\in L_2$, let us denote by $e_1$ the edge on the unique directed path from $s$ to $e$. Then, clearly, 
\begin{align*}
    s_e &= 1-\left(1-x_{e_1}\cdot x_e^{(e_1)}\right)\cdot  \left(1-x_e\cdot x_e^{(e)}\right)\\
    &= 1-\left(1-x_e\right)^2 \le  2x_e\ .
\end{align*}
If $e\in L_3$, let us denote by $A(e)$ the set of ancestor edges of $e$. In the following, we use the inequality $e^{-2x}<1-x$ for $0\le x<1/2$. Then, clearly, 
\begin{align*}
    s_e &= 1-\prod_{e'\in L_1\cap A(e)}\left(1-x_{e'}\cdot x_e^{(e')}\right)\cdot  \prod_{e'\in L_2\cap A(e)}\left(1-x_{e'}\cdot x_e^{(e')}\right)\cdot (1-x_e\cdot x_e^{(e)})\\
    &= 1-(1-x_e)\cdot \left(1-x_e {2\rho m\choose \rho m}^{-1}\right)^{{2\rho m\choose \rho m}} \cdot \prod_{e'\in L_1\cap A(e)}\left(1-x_{e'}\cdot x_e^{(e')}\right)\\
    &\le  1-\exp\left(-4x_e-2\cdot \sum_{e'\in L_1\cap A(e)} x_{e'}\cdot x_e^{(e')}\right)\ .
\end{align*}
The last sum is slightly more tricky to compute. Let $v$ be the endpoint of $e$. We remark that all edges in $\delta^-(v)$ play a symmetric role so that 
\begin{equation*}
     \sum_{e'\in L_1\cap A(e)} x_{e'}\cdot x_e^{(e')}=\sum_{e'\in L_1\cap A(e'')} x_{e'}\cdot x_{e''}^{(e')}\ ,
\end{equation*}
for any $e''\in \delta^-(v)$. To see this, note that the value $x_{e'}$ is the same for all edges in $L_1$. Hence, one only needs to count how many edges $e'\in L_1$ are ancestors of $e''$ and are such that $x_{e''}^{(e')}$ takes a specific value $\delta$. This value $\delta$ is entirely determined by the size of the intersection $|S_{e''}\cap S_{e'}|$ and when summing over all possible edges in $L_1$, we are in fact summing over all sets of size $\rho m$, which makes the symmetry of the construction apparent.
Hence, we can write (recall that $A(v)$ is the set of ancestors of $v$)
\begin{align*}
    \sum_{e'\in L_1\cap A(e)}x_{e'}\cdot x_e^{(e')} &= \frac{1}{|\delta^-(v)|}\cdot \sum_{e''\in \delta^-(v)}\sum_{e'\in L_1\cap A(e'')}x_{e'}\cdot x_{e''}^{(e')}\\
    &= \frac{1}{|\delta^-(v)|}\cdot \sum_{e'\in L_1\cap A(v)}x_{e'} \sum_{e''\in \delta^-(v)} x_{e''}^{(e')}\\
    &= \frac{1}{|\delta^-(v)|}\cdot \sum_{e'\in L_1\cap A(v)}x_{e'} \frac{{(1-\rho )m \choose \rho m}}{{m\choose \rho m}}\\
    &= \frac{1}{|\delta^-(v)|}\cdot \sum_{e'\in L_1\cap A(v)}{(1-\rho) m\choose \rho m}^{-1} \frac{{(1-\rho )m \choose \rho m}}{{m\choose \rho m}}\\
    &=  \frac{1}{|\delta^-(v)|}\cdot \sum_{e'\in L_1\cap A(v)} \frac{1}{{m\choose \rho m}}\\
    &=  \frac{1}{|\delta^-(v)|} = {(1-\rho)m\choose \rho m}^{-1}=x_e\ .
\end{align*}
Therefore 
\begin{equation*}
    s_e \le  1-\exp\left(-4x_e-2\cdot \sum_{e'\in L_1\cap A(e)} x_{e'}\cdot x_e^{(e')}\right) = 1-\exp\left(-6x_e\right)\le 6x_e\ ,
\end{equation*}
where we used the inequality $\exp(x)\ge 1+x$ for all $x$.
\end{proof}
We continue with the following easy, but useful lemma  which intuitively says that in any subtree $S_e$, no edge dominates the total expected out-degree at some vertex $v$.
\begin{lemma}
\label{lem:no_edge_dominates}
    For any $e\in S$, any vertex $v$, and any $e'\in \delta^+(v)$, we have that 
    \begin{equation*}
        \mathbb E[|(\delta^+(v)\setminus \{e'\}) \cap S_e|]\ge (1-o(1))\cdot \mathbb E[|\delta^+(v) \cap S_e|]
    \end{equation*}
\end{lemma}
\begin{proof}
    This is easy to show, it suffices to check the three cases $e\in L_1,e\in L_2,e\in L_3$. We start by the non-trivial case, which is when $e\in L_1$ and $v\in L_2$. In that case, we see by definition of our subtree solutions that 
    \begin{equation*}
        \mathbb E[|\delta^+(v) \cap S_e|]=\sum_{e'\in \delta^+(v)}x_{e'}^{(e)}=1\ ,
    \end{equation*}
    and we have that 
    \begin{align*}
        \mathbb E[|(\delta^+(v)\setminus \{e'\}) \cap S_e|]&\ge  \mathbb E[|(\delta^+(v)) \cap S_e|]-\max_{e'\in \delta^+(v)} x_{e'}^{(e)}\\
        &\ge \mathbb E[|\delta^+(v) \cap S_e|]-\frac{{(1-\rho)m\choose \rho m}}{{m\choose \rho m}}=(1-o(1))\cdot \mathbb E[|(\delta^+(v)) \cap S_e|]\ ,
    \end{align*}
    which concludes this case. In all other cases, we note that for any vertex $v$ and edge $e$, for any $e',e''\in \delta^+(v)$, we have the symmetry that $x_{e'}^{(e)}=x_{e''}^{(e)}$. Moreover, for all vertices $v$, we have that $|\delta^+(v)|=\omega(1)$, which proves the other cases.
\end{proof}
We conclude this part with the last lemma.

\begin{lemma}
\label{lem:conditioning_union_bound}
    For $m$ big enough, and for any event $\mathcal E$ that one round of Sherali-Adams can condition on, for any vertex $v$, we have that 
    \begin{equation*}
        \mathbb E[|\delta^+(v)\cap A|\mid \mathcal E]\ge \Omega(1/m^4)\cdot \mathbb E[|\delta^+(v)\cap A'|\mid \mathcal E]\ ,
    \end{equation*}
\end{lemma}
Before proving this lemma, let us comment on it. The left-hand side is the expected out-degree of a vertex $v$ (after conditioning), while $\mathbb E[|\delta^+(v)\cap A'|\mid \mathcal E]$ refers to expected out-degree estimated by the naive union bound over all subtrees, which counts the edges with multiplicity. Essentially, this lemma states some converse inequality of the union bound, i.e. that at the cost of loosing a factor $m^{O(1)}=(\log n)^{O(1)}$, we can replace the exact value by the union bound estimate. The intuition as to why this is true is that the number of times an edge is selected in a sum of independent binary variables of small expectation, so with superpolynomially small probability it is selected more than say $m^4$ times. Next, we note that the events SA is allowed to condition by have only a polynomially small probability, so even after conditioning, an edge will be selected less than $m^4$ times with high probability.
\begin{proof}
    First, we remark that the number of times an edge $e$ appears in $A'$ is a sum of independent binary random variables, let us denote by $n_e$ this number. Moreover, we can compute that 
    \begin{equation*}
        \mathbb E[n_e] = \sum_{e'\in A(e)} x_{e'}\cdot x_{e}^{(e')}<2\ .
    \end{equation*}
    where we recall that $A(e)$ are the ancestor edges of $e$. To see this, note that if $e\in L_1\cup L_2$, this is easy to check since any such edge has at most $2$ ancestors. If $e\in L_3$, we write 
    \begin{align*}
        \sum_{e'\in A(e)} x_{e'}\cdot x_{e}^{(e')}&=\sum_{e'\in A(e)\cap L_1} x_{e'}\cdot x_{e}^{(e')} + \sum_{e'\in A(e)\cap L_2} x_{e'}\cdot x_{e}^{(e')} + x_e\\
        &\le x_e+ {2\rho m\choose \rho m}\cdot x_e\cdot {2\rho m\choose \rho m}^{-1}\cdot 1 + x_e = 3x_e<2\ .
    \end{align*}
    We recall that we already performed the calculation that $\sum_{e'\in A(e)\cap L_1} x_{e'}\cdot x_{e}^{(e')}=x_e$ in the proof of Lemma \ref{lem:sherali_lifteddistribution1}.
    Using standard Chernoff bounds, for any $t>m^4$ and $m$ big enough, we obtain that 
    \begin{equation*}
    \mathbb P[n_e\ge t]\le \exp(-t/10)\ .
    \end{equation*}
    Second, note that the event $\mathcal E$ must have a probability at least $x_e=\exp(-\Theta(m))$ for some edge $e\in G$ (using Lemma \ref{lem:sherali_lifteddistribution1} again). Therefore, we can write for $t>m^4$ and $m$ big enough,
    \begin{equation*}
    \mathbb P[n_e\ge t\mid \mathcal E]=\frac{\mathbb P[\{n_e\ge t\}\cap \mathcal E]}{\mathbb P[\mathcal E]}\le \frac{\exp(-t/10)}{\exp(-\Theta(m))}\le \exp(-t/20)\ .
    \end{equation*}
    Therefore, we can write 
    \begin{align*}
        \mathbb E[|\delta^+(v)\cap A'|\mid \mathcal E]&=\sum_{e\in \delta^+(v)} \mathbb E[n_e\mid \mathcal E]\\
        &\le  \sum_{e\in \delta^+(v)} \left(\mathbb E[n_e\mid \{n_e\le m^4,\mathcal E\}]+\sum_{t=m^4}^{\infty}\mathbb E[n_e\mid \{n_e= t,\mathcal E\}]\cdot \mathbb P[n_e\ge t\mid \mathcal E]\right)\ .
    \end{align*}
    To analyze the above sum, we need to handle several cases.
    \paragraph{Case 1: $\mathcal E=\{e_1\notin A\}$ and $e=e_1$.} In that case, then we have
    $\mathbb E[n_e\mid \{n_e= t,\mathcal E\}]\cdot \mathbb P[n_e\ge t\mid \mathcal E]=0$ for all $t$ and $\mathbb E[n_e\mid \{n_e\le m^4,\mathcal E\}]=0$. 
    \paragraph{Case 2:  $\mathcal E=\{e_1\in A\}$ and $e= e_1$.} Then in that case, we clearly have that $\mathbb E[n_e\mid \{n_e\le m^4,\mathcal E\}]\ge 1$, and that 
    \begin{equation*}
        \sum_{t=m^4}^{\infty}\mathbb E[n_e\mid \{n_e= t,\mathcal E\}]\cdot \mathbb P[n_e\ge t\mid \mathcal E]\le \sum_{t=m^4}^{\infty} t\cdot \exp(-t/20)=O(1)\ .
    \end{equation*}

    \paragraph{Case 3: $e_1\neq e$.} In that case, our strategy is prove that $\mathbb E[n_e\mid \{n_e\le m^4,\mathcal E\}]\ge n^{-\Theta(1)}$, which will show that the sum $\sum_{t=m^4}^{\infty} t\cdot \exp(-t/20)$ is negligible compared to the total expectation. To see this, we have two subcases.
    \begin{enumerate}
        \item If $e_1\notin D(e)$ then we remark that the event $\{e\in S\}\cap \{e\in S_e\}$ is independent of $\mathcal E$ (also recall that $\mathbb P[e\in S_e]=1$). Then we write
        \begin{align*}
            \mathbb P&[\{e\in S\}\cap \{e\in S_{e}\}\cap \{n_e\le m^4\}\cap \mathcal E]\\
            &=\mathbb P[\{n_e\le m^4\}\mid \mathcal E\cap \{e\in S\}\cap \{e\in S_{e}\}]\cdot \mathbb P[\{e\in S\}\cap \{e\in S_{e}\}\mid \mathcal E]\cdot \mathbb P[\mathcal E]\ge n^{-\Theta(1)}\ ,
        \end{align*}
        where the last inequality uses the fact that the three events $\{e\in S\}, \{e\in S_{e}\}, \mathcal E$ are all independent with probability at least $n^{-\Theta(1)}$ each, and that $\mathbb P[\{n_e> m^4\}\mid \mathcal E\cap \{e\in S\}\cap \{e\in S_{e}\}]\le \exp(-m^4/10)/\mathbb P[\{e\in S\}\cap \{e\in S_{e}\}\cap \mathcal E]\le \exp(-m^3)$.

        This shows that $\mathbb E[n_e\mid \{n_e\le m^4,\mathcal E\}]\ge n^{-\Theta(1)}$.
        \item If $e_1\in D(e)\setminus \{e\}$, then let us consider $e'$ an ancestor of $e$ in $L_1$, for which we remark that $x_e^{(e')}\ge n^{-\Theta(1)}$ and $x_{e_1}^{(e')}=n^{-\Theta(1)}$ by construction. Further, we remark that if we condition by the event $\{e_1\notin S_{e'}\}$, then the events $\{e'\in S\}$ and $\mathcal E$ become independent. Therefore, we write
        \begin{align*}
            &\mathbb P[\{e'\in S\}\cap \{e_1\notin S_{e'}\} \cap \{e\in S_{e'}\}\cap \{n_e\le m^4\}\cap \mathcal E]\\
           &= \mathbb P[\{e'\in S\} \cap \{e\in S_{e'}\}\cap \{n_e\le m^4\}\cap \mathcal E\mid  \{e_1\notin S_{e'}\}]\cdot \mathbb P[\{e_1\notin S_{e'}\}]\\
           &\ge n^{-\Theta(1)}\ ,
           \end{align*}
           where the last inequality uses the fact that $\mathbb P[\{e_1\notin S_{e'}\}]=1-x_{e_1}^{(e')}>1-n^{-\Theta(1)}$, the fact that $\mathbb P[\{n_e> m^4\}\mid \mathcal E\cap \{e'\in S\}\cap \{e_1\notin S_{e'}\}\cap \{e\in S_{e'}\}]\le \exp(-m^3)$, and the fact that $\mathbb P[\{e'\in S\} \cap \{e\in S_{e'}\}\cap  \mathcal E\mid  \{e_1\notin S_{e'}\}]\ge n^{-\Theta(1)}$.

           The second fact can be deduced from the third by a similar calculation as in the previous subcase. The last fact is not obvious. Note that after conditioning by $\mathcal \{e_1\notin S_{e'}\}$ the three events $\mathcal E,\{e'\in S\},\{e\in S_{e'}\}$ become independent. Hence it suffices to show that $\mathbb P[ \mathcal E\mid  \{e_1\notin S_{e'}\}]\ge n^{-\Theta(1)}$. To this end, if $\mathcal E$ is a negative event, we clearly have that
           \begin{equation*}
               \mathbb P[ \mathcal E\mid  \{e_1\notin S_{e'}\}]\ge \mathbb P[\mathcal E]\ge 1-6x_{e_1}>1/2\ ,
           \end{equation*}
           where we used Lemma \ref{lem:sherali_lifteddistribution1}. Otherwise, if $\mathcal E$ is a positive event, then we have that 
            \begin{equation*}
               \mathbb P[ \mathcal E\mid  \{e_1\notin S_{e'}\}]= \mathbb P[ \{e_1\in A\}\mid  \{e_1\notin S_{e'}\}]\ge \mathbb P[ \{e_1\in S\}\mid  \{e_1\notin S_{e'}\}]=x_{e_1}=n^{-\Theta(1)}\ .
           \end{equation*}
    \end{enumerate}
     To conclude, we were able to prove that in all cases, we have for any edge $e\in \delta^+(v)$ that, either $\mathcal E=\{e\notin A\}$, or that 
     \begin{align*}
         \sum_{t=m^4}^{\infty}\mathbb E[n_e\mid \{n_e= t,\mathcal E\}]\cdot \mathbb P[n_e\ge t\mid \mathcal E]&\le  \sum_{t=m^4}^{\infty}t\cdot \exp(-t/20)\\
         &\le O(\exp(-m^2))\\
         &\le O(\mathbb E[n_e\mid \{n_e\le m^4,\mathcal E\}])\ .
     \end{align*}
     Therefore, we can conclude 
         \begin{align*}
        \mathbb E[|\delta^+(v)\cap A'|\mid \mathcal E]&=\sum_{e\in \delta^+(v)} \mathbb E[n_e\mid \mathcal E]\\
        &\le  \sum_{e\in \delta^+(v)} \mathbb E[n_e\mid \{n_e\le m^4,\mathcal E\}]+\sum_{t=m^4}^{\infty}\mathbb E[n_e\mid \{n_e= t,\mathcal E\}]\cdot \mathbb P[n_e\ge t\mid \mathcal E]\\
        &\le O\left(\sum_{e\in \delta^+(v)} \mathbb E[n_e\mid \{n_e\le m^4,\mathcal E\}]\right)\le O(m^4) \cdot  \mathbb E[|\delta^+(v)\cap A|\mid \mathcal E]\ .
    \end{align*}
\end{proof}

\subsubsection{Verifying the covering constraints}

For convenience, we recall here the covering constraint.
\begin{equation*}
    \sum_{e\in \delta^+(v)}x_e\ge k_v\cdot \sum_{e\in \delta^-(v)} x_e
\end{equation*} for any vertex $v$ which is not a sink. We use the convention that $\sum_{e\in \delta^-(v)} x_e=1$ if $v=s$. We will verify those constraints up to a factor of $(\log n)^{O(1)}$.

Using Lemma \ref{lem:sherali_lifteddistribution1}, it is clear that any constraint is satisfied up to a multiplicative factor of $6$ if we do not condition by any event. If we perform some conditioning, we use Lemma \ref{lem:conditioning_union_bound} to argue that we can count all edges with their multiplicity, at the cost of loosing some factor $(\log n)^{O(1)}$ (and this is true even after conditioning). 

We have now a simple case analysis.
\paragraph{Case 1: $v=s$.} This is case is easy, as $\mathcal E$ can be dependent of at most ${2\rho m\choose \rho m}$ edges in $\delta^+(s)$ (by the structure of the graph). Since $s$ has ${m \choose \rho m}$ outgoing edges, all playing a symmetric role, the influence of $\mathcal E$ is negligible.

\paragraph{Case 2: $v\neq s$ and $\mathcal E=\{e_1\notin A\}$ is a negative event.} Then we condition by the outcome of the first step of the sampling (i.e. the set selected as the shadow set). After this conditioning, we see that by linearity of expectation, and by the fact that the all the subtree solutions satisfy the covering constraints of the assignment LP, we will have that the covering constraint is satisfied in expectation by our sampling procedure. There is one slight caveat however, which is if $e_1\in \delta^+(v)$, in which case the out-degree of $v$ might decrease in expectation. However, by Lemma \ref{lem:no_edge_dominates}, this loss in negligible in any subtree $S_e$ for any edge $e\in E$. Hence by linearity of the expectation, the loss in negligible for any outcome for the shadow set.

Formally, we write, for any possible outcome $X$ of the shadow set $S$:
\begin{align*}
    \mathbb E\left[|\delta^+(v)\cap A|\mid \mathcal E \right]&\ge \Omega(1/m^4)\cdot  \mathbb E\left[|\delta^+(v)\cap A'|\mid \mathcal E \right] \\
    &=\Omega(1/m^4)\sum_{X}\sum_{e\in X}\mathbb E\left[|(\delta^+(v)\setminus \{e_1\})\cap S_{e}|\right]\cdot \mathbb P[\{S=X\}\mid \mathcal E]\\
    &\ge \Omega((1-o(1))/m^4)\sum_{X}\sum_{e\in X}\mathbb E\left[|\delta^+(v)\cap S_{e}|\right]\cdot \mathbb P[\{S=X\}\mid \mathcal E]\\
    &\ge \Omega(1/m^4)\sum_{X}\sum_{e\in X}k_v\cdot \mathbb E\left[|\delta^-(v)\cap S_{e}|\right]\cdot \mathbb P[\{S=X\}\mid \mathcal E]\\
    &\ge  \Omega(1/m^4)\cdot k_v\cdot \mathbb E\left[|\delta^-(v)\cap A|\mid \mathcal E \right]\ .
\end{align*}
\paragraph{Case 3: $v\neq s$ and  $\mathcal E=\{e_1\in A\}$ is a positive event.} Similarly, we use Lemma \ref{lem:conditioning_union_bound} to argue with union bounds only, at the cost of loosing a factor $(\log n)^{O(1)}$. The proof is very similar to Case 2, using the conditioning by every possible outcome for the shadow set. However, there is one crucial difficulty, which shows up when $e_1\in \delta^-(v)$. In that case, it is crucial that the expected out-degree at vertex $v$ is at least $\Omega(k_v)$, because the event $\mathcal E$ adds an additive $1$ to the in-degree of $v$. Fortunately, using the properties of our distributions, and Lemma \ref{lem:sherali_lifteddistribution1}, we can write
\begin{align*}
     \mathbb P[\{e_1\in S\}\mid \{e_1\in A\}] &= \frac{\mathbb P[\{e_1\in A\}\mid \{e_1\in S\}]\cdot \mathbb P[\{e_1\in S\}]}{\mathbb P[\{e_1\in A\}]}\\
     &\ge (1/6)\cdot \mathbb P[\{e_1\in A\}\mid \{e_1\in S\}] = 1/6\ .
\end{align*}
Hence, with constant probability (after conditioning), the edge $e_1$ is also in the shadow set $S$, which will trigger a subtree rooted at $v$ and ensure an additional expected out-degree $k_v$ at vertex $v$ in expectation (after conditioning).
\subsubsection{Verifying the packing constraints}

For convenience, we recall here the packing constraint.
\begin{equation*}
    \sum_{e\in \delta^-(v)}x_e\le 1
\end{equation*} for any vertex $v$. We will verify those constraints up to a multiplicative factor of $O(1)$.

Using Lemma \ref{lem:sherali_lifteddistribution1}, it is clear that any constraint is satisfied up to a multiplicative factor of $6$ if we do not condition by any event. 

Otherwise, let $e_1$ be the edge considered in the conditioning, and let $\mathcal E$ be the corresponding event. Let us do a case analysis. 

\paragraph{Case 1: $\mathcal E=\{e_1\notin A\}$.} We notice that after conditioning, the expected in-degree on vertex $v$ in the set $A$ can only decrease. To see this formally, note that for any $e$,
  \begin{equation*}
        \mathbb P[\{e\in S\}\mid \{e_1\notin A\}]=\frac{\mathbb P[\{e_1\notin A\}\mid \{e\in S\}]\cdot \mathbb P[\{e\in S\}]}{\mathbb P[\{e_1\notin A\}]}\le \mathbb P[\{e\in S\}]\ ,
    \end{equation*}
where clearly conditioning by $\{e\in S\}$ cannot decrease the probability of the event $\{e_1\in A\}$. Second, note that the in-degree cannot be negatively correlated with any event of the form $\{e\in S\}$. By Lemma \ref{lem:sherali_lifteddistribution1}, the expected in-degree was at most $6$ to start with. Using previous remarks, it implies that the constraint is satisfied up to a factor of $6$ after the conditioning.

\paragraph{Case 2. $\mathcal E=\{e_1\in A\}$.} Let $v$ be the vertex at which the packing constraint is. The strategy is to first understand how the probability of each edge being in the shadow set changes after conditioning by the event $\mathcal E$, then argue about the congestion that each edge in the shadow set will induce on vertex $v$ by triggering a subtree.

If $v\in L_1$, then $\delta^-(v)$ contains only one edge, so the constraint is satisfied deterministically. 

If $v\in L_2$, let us assume in the worst-case that all edges in $e\in L_1\cap A(v)$ have probability 1 of appearing in the shadow set (after conditioning by $\mathcal E$). Using our definitions for the subtree solutions, the expected congestion induced by those triggered subtrees is then equal to at most 
\begin{equation*}
    \sum_{e\in L_1\cap A(v)} {2\rho m\choose \rho m}^{-1} = 1\ .
\end{equation*}
Hence we only need to worry about the probability of the edges $e\in \delta^-(v)$ being selected in the shadow set. Note that if $e_{1}\notin \delta^+(v)\cup \delta^-(v)$, the event $\{e\in S\}$ is independent of $\mathcal E$ for any $e\in \delta^-(v)$ so we are done. Otherwise, if $e_1\in \delta^-(v)$, then only the probability of $e_1$ being in $S$ changes. Even after conditioning, this probability is at most 1 anyway so we are done. If $e_1\in \delta^+(v)$, for any edge $e\in \delta^-(v)$ we write
\begin{align*}
     \mathbb P[\{e\in S\}\mid \{e_1\in A\}] &= \frac{\mathbb P[\{e_1\in A\}\mid \{e\in S\}]\cdot \mathbb P[\{e\in S\}]}{\mathbb P[\{e_1\in A\}]}\\
     &\le \frac{1\cdot x_e}{x_{e_1}}={2\rho m\choose \rho m}^{-1}\ ,
\end{align*}
where the inequality uses Lemma \ref{lem:sherali_lifteddistribution1}. Hence, each edge $e\in \delta^-(v)$ induces an expected congestion of ${2\rho m\choose \rho m}^{-1}$ on vertex $v$, after conditioning on event $\mathcal E$. There are exactly ${2\rho m\choose \rho m}$ such edges, which concludes that the total congestion is $O(1)$ in the case where $v\in L_2$.

Finally, $v\in L_3$ is the most tricky subcase. If $e_1\in L_1$, then the probability of some edge $e$ being in the shadow set is only modified after conditioning if $e=e_1$. This is only one edge, so even if the probability rises to 1 this is not an issue as the subtree $S_e$ has an expected congestion on $v$ of $o(1)$. If $e_1\in L_2$, the probability of at most $2$ edges of being in $S$ can increase after conditioning by $\mathcal E$, and the same argument applies. 

Finally, we are left with the case that $e_1=(u',v')\in L_3$. Note that there is at most one edge $(u',v)$ if it exists, so we can ignore all the edges in $\delta^-(u')$ by loosing a additive $1$ on the congestion. Note that all edges in $\delta^-(v)$ appears in $S$ independently of $\mathcal E$, except for possibly one edge if $e_1\in \delta^-(v)$. Hence we only need to worry about the probability of the edges in $L_1$ being selected in the shadow set, after conditioning by the event $\mathcal E$. For any edge $e\in  L_1$, we write

\begin{align*}
     \mathbb P[\{e\in S\}\mid \{e_1\in A\}] &= \frac{\mathbb P[\{e_1\in A\}\mid \{e\in S\}]\cdot \mathbb P[\{e\in S\}]}{\mathbb P[\{e_1\in A\}]}\\
     &\le \frac{\mathbb P\left[\bigcup_{e'\in E}\{\{e'\in S\}\cap \{e_1\in S_{e'}\}\}\mid \{e\in S\}\right]\cdot \mathbb P[\{e\in S\}]}{\mathbb P[\{e_1\in A\}]}\\
     &\le  \left(\mathbb P\left[\{e_1\in A\}\right]+\mathbb P\left[\{e_1\in S_{e}\}\right]\right)\cdot \frac{\mathbb P[\{e\in S\}]}{\mathbb P[\{e_1\in A\}]}\\
     &\le  \left(6x_{e_1}+x_{e_1}^{(e)}\right)\cdot \frac{\mathbb P[\{e\in S\}]}{\mathbb P[\{e_1\in A\}]}\\
     &\le \left(6x_{e_1}+x_{e_1}^{(e)}\right)\cdot \frac{\mathbb P[\{e\in S\}]}{\mathbb P[\{e_1\in S\}]}\\
     &=6x_{e_1}+x_{e_1}^{(e)}\ .
\end{align*}
where we used a union bound in the third line, Lemma \ref{lem:sherali_lifteddistribution1} in the fifth, and the fact that $x_e=x_{e_1}$ in our instance in the last line. Therefore, the expected congestion induced by edges in $L_1$ conditioned on $\mathcal E$ is equal to at most
\begin{align*}
    &\sum_{e\in L_1}(6x_{e_1}+ x_{e_1}^{(e)})\cdot \sum_{e'\in \delta^-(v)}x_{e'}^{(e)}\\
    &=  \sum_{e\in L_1}6x_{e}\cdot \sum_{e'\in \delta^-(v)}x_{e'}^{(e)} + \sum_{e\in L_1}x_{e_1}^{(e)}\cdot  \sum_{e'\in \delta^-(v)}x_{e'}^{(e)}\\
    &=  \sum_{e\in L_1}6{(1-\rho)m\choose \rho m}^{-1}\cdot \frac{{(1-\rho)m\choose \rho m}}{{m\choose \rho m}} +\sum_{e\in L_1}x_{e_1}^{(e)} \cdot \frac{{(1-\rho)m\choose \rho m}}{{m\choose \rho m}}\\
    &= 6+\sum_{e\in L_1\cap A(u')}\frac{{(1-\rho)m\choose \rho m}}{{m\choose \rho m}\cdot {m-2\rho m+|S_e\cap S_{e_1}|\choose |S_e\cap S_{e_1}|}} \cdot \frac{{(1-\rho)m\choose \rho m}}{{m\choose \rho m}}\\
    &= O(1)+\left(\frac{{(1-\rho)m\choose \rho m}}{{m\choose \rho m}}\right)^2 \cdot \sum_{e\in L_1\cap A(u')}{m-2\rho m+|S_e\cap S_{e_1}|\choose |S_e\cap S_{e_1}|}^{-1}\\
    &= O(1)+\left(\frac{{(1-\rho)m\choose \rho m}}{{m\choose \rho m}}\right)^2 \cdot \sum_{j=0}^{\rho m} {m-2\rho m+j\choose j}^{-1}{\rho m\choose j}^{2}
\end{align*}
To analyze this last sum, we need that
$${m-2\rho m+j\choose j}^{-1}{\rho m\choose j}^{2}\cdot \left(\frac{{(1-\rho)m\choose \rho m}}{{m\choose \rho m}}\right)^2\le m^{O(1)}$$ for all the range of possible values of $j$. Indeed, using Lemma \ref{lem:binomial_coeff}, with $j=x \rho m$ for $0\le x \le 1$, we have
\begin{equation*}
    \log_2\left({m-2\rho m+j\choose j}^{-1}{\rho m\choose j}^{2}\right)/m\le (2\rho)\cdot  h\left(x\right)-(1-2\rho+x \rho)\cdot h\left(\frac{x \rho}{1-2\rho +x \rho}\right)+O(\log (m)/m)\ .
\end{equation*}
To conclude, we study the function 
\begin{equation*}
    f:x\mapsto (2\rho)\cdot  h\left(x\right)-(1-2\rho+x \rho)\cdot h\left(\frac{x \rho}{1-2\rho +x \rho}\right)+2\cdot \left((1-\rho)h\left(\frac{\rho}{1-\rho}\right)-h(\rho)\right)
\end{equation*}
One can verify that 
\begin{equation*}
    f'(x)=\rho \cdot \left(2\log_2(1-x)-2\log_2(x)+\log_2\left(\frac{x \rho}{1-2\rho+x\rho}\right)\right)
\end{equation*}
which is non-negative if and only if $x\le \rho$. Therefore, this function admits a maximum at $x=\rho$, where 
\begin{equation*}
    f(\rho)=2\rho h(\rho)-(1-\rho)^2\cdot h\left(\left(\frac{\rho}{1-\rho}\right)^2\right)+2(1-\rho)h\left(\frac{\rho}{1-\rho}\right)-2h(\rho)\ .
\end{equation*}
Finally, one can expand $f(\rho)$ for $\rho\rightarrow 0^+$ to obtain that 
\begin{equation*}
    f(\rho)=-\frac{\rho^2}{\log(2)}+O(\rho^3)\ .
\end{equation*}
Hence for any constant $\rho>0$ small enough, we obtain that $f(\rho) < 0$.

Hence, we obtain 
\begin{align*}
    \sum_{e\in L_1}(x_{e_1}+ x_{e_1}^{(e)})\cdot \sum_{e'\in \delta^-(v)}x_{e'}^{(e)}\le  1+o(1)=O(1)\ ,
\end{align*}
for any constant $\rho$ small enough. This concludes the last case.

\subsubsection{Connecting the dots}
We were able to verify all constraints with all possible conditioning, up to a multiplicative factor of $(\log n)^{O(1)}$, assuming $\rho$ is a small enough constant, and $m$ is big enough. Hence, by Theorem \ref{thm:sherali_distributions}, we obtain that this instance is feasible for $1$ round of Sherali-Adams with an approximation rate $(\log n)^{O(1)}$. But we prove in Section \ref{sec:construction} that any integral solution must have approximation factor $n^{\Omega(1)}$, which concludes the proof of Theorem \ref{thm:sherali_1round_gap}.
\section{The general construction}
\label{sec:construction}
In this part, we describe our instances (for any depth) and discuss some of their properties. In particular, this section contains a proof that our instances do not contain any $n^{o(1/\ell)}$-approximate integral solutions (Lemma \ref{lem:no_integral_sol}), and a proof of the first three properties of Theorem \ref{thm:locality_gap}. We will denote by $\ell$ the depth of the instance, which will be parameterized by its size $n$, its depth $\ell$, and a constant $\rho$. We will denote by $G^{(\rho)}_{n,\ell}$ the instance obtained in that way. The instance used in Section \ref{sec:sherali1round} is the special case where $\ell=3$.

\subsection{The vertex set}

We have a ground set $\mathcal U=[m]$ of size $m$ for $m$ big enough. We fix a constant $\rho$ which will be taken small enough. We also use the parameter $\epsilon$ such that $\ell=3/\epsilon$. We assume that $\epsilon m=\omega(\log m)$ (or equivalently, $\ell=o(\log n/\log \log n)$), although this is not strictly needed for the construction itself, we will use this assumption later in some proofs. We choose $1/\epsilon$ and $m$ big enough so that $\epsilon \rho m\in \mathbb N$, $1/\epsilon\in \mathbb N$. For $n$ big enough, this will always be possible while modifying the instance size $n$ and the number of layers by some factor $O(1)$ only. This will not affect the asymptotic behavior of our results.

The graph will contain $\ell+1$ layers $L_0,L_1,\ldots ,L_i,\ldots, L_{\ell=3/\epsilon}$. Until $i\le 2/\epsilon$, we will have an \textit{expanding phase} where the sets labelling the vertices get bigger and bigger. Then, from layer $L_{2/\epsilon}$ to the last layer, the labelling sets get smaller and smaller; this is the \textit{collapsing phase}. With this intuition in mind, we are ready to describe the instance.

\paragraph{Source vertex.} In layer $L_0$, there is only the \textit{source} vertex $s$, which we label with the empty subset, i.e. we set $S_s=\emptyset$.

\paragraph{Expanding phase.} In layer $L_i$ for $i< 2/\epsilon$, there is one vertex $v$ for each set $S_v\subseteq \mathcal U$ of size $i\epsilon \rho m$. Hence, we have  
\begin{equation*}
    |L_i|={m\choose i\epsilon \rho m}\ ,
\end{equation*}
for $i< 2/\epsilon$.

\paragraph{Collapsing phase.} In layer $L_i$ for $2/\epsilon \le i\le 3/\epsilon$, there is one vertex $v$ for each set $S_v\subseteq \mathcal U$ of size $4\rho m- i\epsilon \rho m$. Hence, we have that 
\begin{equation*}
    |L_i|={m\choose 4\rho m- i\epsilon \rho m}\ ,
\end{equation*}
for $i\ge 2/\epsilon$.

\paragraph{Sinks.} The \textit{sinks} are all the vertices in the last layer (layer $i=3/\epsilon$).

\subsection{The edge set}

For each $i$, we define the set $E_i$ of edges having their endpoint in $L_i$ as follows:
\begin{equation*}
    E_i:=\begin{cases}
            \left\lbrace(u,v)\ \vert \  u\in L_{i-1},v\in L_{i} \mbox{ and }S_u\subseteq S_v\right\rbrace \mbox{ if } 1\le i\le 2/\epsilon,\\
             \left\lbrace(u,v)\ \vert \  u\in L_{i-1},v\in L_{i} \mbox{ and }S_v\subseteq S_u\right\rbrace \mbox{ if } 2/\epsilon < i\le 3/\epsilon.
        \end{cases}
\end{equation*}
Then, we set the total edge set $E$ as
\begin{equation*}
    E:=\bigcup_{1\le i\le 3/\epsilon} E_i\ .
\end{equation*}
In other words, there can exist edges only in-between consecutive layers. Moreover, any pair of vertices $(u,v)$ belonging to consecutive layers are connected by an edge if and only if the labels $S_u,S_v$ have an inclusion relationship (i.e. either $S_u\subseteq S_v$ or $S_v\subseteq S_u$).

\subsection{Degree requirements}

In this part, we describe the required out-degree for each vertex $v$ in the arborescence to be covered, which we denote by $k_v$. For ease of notation, we will also denote 
\begin{equation*}
    \gamma_v:=k_v/|\delta^+(v)|\ 
\end{equation*} the ratio between required out-degree for vertex $v$, and the out-degree of $v$ in the graph. All these quantities will only depend on the layer that $v$ is in, and not specifically on which vertex $v$ is. Hence, we will have quantities $k_i,\gamma_i,\delta^+_i$ when referring to these quantities for any vertex $v\in L_i$.

We set the $k_i$s such that in an integral solution $T$, the following must hold:
\begin{enumerate}
    \item $|T\cap L_{1/\epsilon}|=\frac{{m\choose \rho m}}{{(1-\rho)m\choose \rho m}}$,
    \item $|T\cap L_{2/\epsilon}|=\frac{{m\choose \rho m}}{{2\rho m\choose \rho m}}$, and
    \item $|T\cap L_{3/\epsilon}|={m\choose \rho m}$.
\end{enumerate}
To achieve this, we set the the $\gamma_i$s as follows:
\begin{equation*}
\gamma_i:=
    \begin{cases}
        \gamma'_1\mbox{ if } 0\le i< 1/\epsilon\ ,\\
        \gamma'_2 \mbox{ if } 1/\epsilon\le i< 2/\epsilon\ ,\\
        \gamma'_3 \mbox{ if } 2/\epsilon\le i< 3/\epsilon\ ,
    \end{cases}
\end{equation*}
where
\begin{equation*}
    \gamma'_1 := \frac{(\epsilon \rho m)!}{\left((\rho m)!\cdot {(1-\rho) m\choose \rho m}\right)^{\epsilon}}\ ,
\end{equation*}
\begin{equation*}
    \gamma'_2 := \frac{(\epsilon \rho m)!}{\left((\rho m)!\cdot {2\rho m\choose \rho m}\right)^{\epsilon}}\ ,\mbox{ and}
\end{equation*}
\begin{equation*}
    \gamma'_3 := \frac{(\epsilon \rho m)!}{\left((\rho m)!\right)^{\epsilon}}\ ,
\end{equation*}

Let us verify that this satisfies our three desiderata. In an integral solution $T$ satisfying the above degree requirements, we must have that 
\begin{align*}
    |T\cap L_{1/\epsilon}|&=\prod_{i=0}^{1/\epsilon-1}(\gamma_i \delta_i^{+})\\
    &=(\gamma'_1)^{1/\epsilon}\cdot \prod_{i=0}^{1/\epsilon-1}\delta_i^{+}\\
    &=(\gamma'_1)^{1/\epsilon}\cdot \prod_{i=0}^{1/\epsilon-1}{m-i\epsilon \rho m\choose \epsilon\rho m}\\
    &= \frac{1}{(\rho m)!\cdot {(1-\rho) m\choose \rho m}}\cdot \prod_{i=0}^{1/\epsilon-1}\left[(\epsilon \rho m)!\cdot {m-i\epsilon \rho m\choose \epsilon\rho m}\right]\\
    &= \frac{1}{(\rho m)!\cdot {(1-\rho) m\choose \rho m}}\cdot \frac{m!}{((1-\rho)m)!} = \frac{{m\choose \rho m}}{{(1-\rho)m\choose \rho m}}\ .
\end{align*}
Similarly, we must have that 

\begin{align*}
    |T\cap L_{2/\epsilon}|&=\prod_{i=0}^{2/\epsilon-1}(\gamma_i \delta_i^{+})\\
    &=\frac{{m\choose \rho m}}{{(1-\rho)m\choose \rho m}} \cdot \frac{1}{(\rho m)!\cdot {2\rho m\choose \rho m}} \cdot \prod_{i=1/\epsilon }^{2/\epsilon-1}\left[(\epsilon \rho m)!\cdot {m-i\epsilon \rho m\choose \epsilon \rho m}\right]\\
    &=\frac{{m\choose \rho m}}{{(1-\rho)m\choose \rho m}} \cdot \frac{1}{(\rho m)!\cdot {2\rho m\choose \rho m}} \cdot \frac{((1-\rho )m)!}{((1-2\rho)m)!}\\
    &= \frac{{m\choose \rho m}}{{(1-\rho)m\choose \rho m}} \cdot \frac{1}{ {2\rho m\choose \rho m}} \cdot {(1-\rho)m\choose \rho m}=\frac{{m\choose \rho m}}{{2\rho m\choose \rho m}}\ .
\end{align*}
By a similar calculation (note that from layers $2/\epsilon$ to $3/\epsilon$ we are in the collapsing phase, so the expression of $\delta_i^+$ changes), we also have that 
\begin{align*}
    |T\cap L_{3/\epsilon}|&=\prod_{i=0}^{3/\epsilon-1}(\gamma_i \delta_i^{+})\\
    &=\frac{{m\choose \rho m}}{{2\rho m\choose \rho m}} \cdot \frac{1}{(\rho m)!} \cdot \prod_{i=2/\epsilon }^{3/\epsilon-1}\left[(\epsilon \rho m)!\cdot {4\rho m-i\epsilon \rho m\choose \epsilon \rho m}\right]\\
    &=\frac{{m\choose \rho m}}{{2\rho m\choose \rho m}} \cdot \frac{1}{(\rho m)!} \cdot \frac{(2\rho m)!}{(\rho m)!}={m\choose \rho m}\ .
\end{align*}

\subsection{Required out-degrees are large}
We prove the following lemma.
\begin{lemma}
\label{lem:large_required_degrees}
    There exists an absolute constant $\xi>0$ such that for any $\rho<\xi$, any $m$ big enough and $\epsilon$ small enough, we have that $k_u\ge n^{\Omega(1/\ell)}=n^{\Omega(\epsilon)} $ for any vertex $u$. Moreover, for $\ell=3$ (which means $\epsilon=1$), then $k_u=n^{\Omega(1)}$ for all $u$.
\end{lemma}
\begin{proof}
The second case is easy to check, remembering that $n=2^{O(m)}$ and that if $\epsilon=1$, we have only $3$ different values of $k$, $k_0=\frac{{m\choose \rho m}}{{(1-\rho)m\choose \rho m}}$, $k_1=\frac{{(1-\rho)m\choose \rho m}}{{2 \rho m\choose \rho m}}$, and $k_2={2 \rho m\choose \rho m}$ (recall that $\rho$ is considered a constant).

    In the first case, we use Stirling's formula, the fact that $\epsilon m=\omega(\log m)$, and that $\epsilon$ is a small enough constant; and we obtain for all $i<1/\epsilon$,
\begin{equation*}
    k_i\ge \gamma'_1 \delta^+_{1/\epsilon} \sim m^{\Theta(1)}\cdot {(1-\rho) m \choose \rho m}^{-\epsilon}\cdot \frac{(\epsilon\rho m)^{\epsilon \rho m}}{(\rho m)^{\epsilon \rho m}} \cdot {m-\rho m \choose \epsilon \rho m}\ ,
\end{equation*}
hence we get
\begin{equation*}
    \log_2(k_i)/m\ge -\epsilon (1-\rho)h\left(\frac{\rho}{1-\rho}\right)-\log_2(1/\epsilon)\epsilon \rho+(1-\rho) h\left(\frac{\epsilon \rho}{(1-\rho)}\right)+O(\log (m)/m)\ ,
\end{equation*}
where we recall that $h$ is the entropy function of the Bernoulli distribution. Then, we use the expansion $h(x)= x\frac{(1-\log(x))}{\log(2)}-\frac{x^2}{\log (4)}+O(x^3)$ at $x=0$ ($\log$ is the natural logarithm) to obtain

\begin{align*}
    \log_2(k_i)/m &\ge -\epsilon (1-\rho)h\left(\frac{\rho}{1-\rho}\right)-\log_2(1/\epsilon)\epsilon \rho+(1-\rho) h\left(\frac{\epsilon \rho}{(1-\rho)}\right)+O(\log (m)/m)\\
    \ge -\epsilon (1-\rho)&h\left(\frac{\rho}{1-\rho}\right)-\log_2(1/\epsilon)\epsilon \rho+(1-\rho)\cdot \frac{\epsilon \rho (1-\log(\epsilon \rho /(1-\rho)))}{(1-\rho)\log(2)} +O(\epsilon^2)+O(\log (m)/m)\\
    \ge -\epsilon \cdot &\left(\rho \frac{1-\log(\rho/(1-\rho))}{\log(2)}-\frac{\rho^2}{(1-\rho)\log(4)}+O(\rho^3)\right) + \epsilon \rho \frac{1-\log(\rho/(1-\rho))}{\log(2)} + o(\epsilon)\\
    &\ge \epsilon \left(\frac{\rho^2}{(1-\rho)\log(4)}+O(\rho^3)\right)+o(\epsilon)\ .
\end{align*}
Therefore, for $\rho$ a small enough constant, we obtain that $\log_2(k_i)/m\ge \Omega(\epsilon)$, hence $k_i\ge n^{\Omega(\epsilon)}=n^{\Omega(1/\ell)}$, for any $i<1/\epsilon$. The other cases are similar: for $1/\epsilon \le i<2/\epsilon$, we write

\begin{equation*}
    k_i\ge \gamma'_2 {m-2\rho m \choose \epsilon \rho m}\sim m^{\Theta(1)}\cdot {2\rho m \choose \rho m}^{-\epsilon}\cdot \frac{(\epsilon\rho m)^{\epsilon \rho m}}{(\rho m)^{\epsilon \rho m}} \cdot {m-2\rho m \choose \epsilon \rho m}\ .
\end{equation*}
Therefore we get
\begin{align*}
    \log_2(k_i)/m&\ge -2\epsilon \rho-\log_2(1/\epsilon)\epsilon \rho+(1-2\rho) h\left(\epsilon \rho/(1-2\rho)\right)+O(\log (m)/m)\\
    & =-2\epsilon \rho  -\log_2(1/\epsilon)\epsilon \rho + \epsilon \rho\frac{1-\log(\epsilon \rho /(1-2\rho))}{\log(2)}+o(\epsilon)\\
    &= -2\epsilon \rho  -\log_2(1/\epsilon)\epsilon \rho + \epsilon \rho\frac{1-\log(\epsilon)-\log( \rho)+\log(1-2\rho)}{\log(2)}+o(\epsilon)\\
    &= \rho \epsilon \left(-2+\frac{1-\log (\rho)+\log(1-2\rho)}{\log(2)}\right)+o(\epsilon)= \rho \epsilon \log_2(1/\rho)+O(\rho \epsilon)+o(\epsilon)\ .
\end{align*}
Similarly, for $\rho$ a small enough constant, we obtain the desired result. Finally, for $i\ge 2/\epsilon$, we obtain
\begin{equation*}
   k_i\ge  m^{\Theta(1)}\cdot (\epsilon)^{\epsilon \rho m}\cdot {\rho m\choose  \epsilon \rho m}
\end{equation*} and taking the log then dividing by $m$ we obtain
\begin{align*}
    \log_2 (k_i)/m &= -\epsilon \rho \log_2(1/\epsilon) + \rho h(\epsilon)+O(\log(m)/m)\\
    &= -\epsilon \rho \log_2(1/\epsilon) + \rho h(\epsilon)+o(\epsilon)\\
    &= -\epsilon \rho \log_2(1/\epsilon) + \rho \epsilon\frac{1-\log(\epsilon)}{\log(2)}+o(\epsilon)\\
    &= \frac{\rho}{\log(2)}\epsilon + o(\epsilon)\ge 1.44 \rho \epsilon +o(\epsilon)\ ,
\end{align*}
which concludes the last case, and the proof.
\end{proof}
\subsection{Number of paths in the instance}

In this part, we verify that every vertex belongs to a superpolynomial number of paths.

\begin{lemma}
\label{lem:number_paths}
    For any vertex $v\in G_{n,\ell}^{(\rho)}$, there are at least $n^{\Omega(\log(\ell))}$ directed paths starting or ending at $v$.
\end{lemma}
\begin{proof}
    The proof is easy by noticing that the minimum out-degree $\delta^+_{\text{min}}$ of a vertex which is not a sink satisfies
    \begin{equation*}
        \delta^+_{\text{min}}=\min_{0\le i\le \ell-1}\delta_i^+\ge {\rho m\choose \epsilon \rho m}\ .
    \end{equation*}
    Similarly, the minimum in-degree of any vertex at distance at least $1/\epsilon$ from the source is equal to
    \begin{equation*}
        \delta^-_{\text{min}}=\min_{1/\epsilon\le i\le \ell-1}\delta^-_i\ge {\rho m\choose \epsilon \rho m}\ .
    \end{equation*}
    By a standard expansion of the entropy function, we get that 
    \begin{equation*}
        \log_2(\delta^+_{\text{min}})/m \ge h(\epsilon)+O(\log m / m) \ge \Omega(\epsilon \log(1/\epsilon))+O(\epsilon)+O(\log m/ m) = \Omega(\epsilon \log(1/\epsilon))\ , 
    \end{equation*}
    and the same for $\delta^-_{\text{min}}$.
   Hence if $v$ is at distance at least $1/\epsilon$ from the source, there at least $2^{1/\epsilon \Omega(m \epsilon \log(1/\epsilon)}=n^{\Omega( \log(\ell))}$ directed paths of length $\Omega(\ell)$ ending at $v$. Otherwise, $v$ is at distance at least $\Omega(\ell)$ from the sinks and we conclude similarly with our bound on $\delta^+_{\text{min}}$.
\end{proof}

\subsection{Ruling out integral solutions}
This last part is the most tricky of this whole section, and also the most important. In the following lemma, our instance of depth $3$ is the case where $\ell=3$ and $\epsilon=1$.
\begin{lemma}
\label{lem:no_integral_sol}
    There exists some absolute constant $\xi>0$ such that for any $\rho<\xi$, if $T$ is an $\alpha$-approximate integral solution in the constructed instance $G_{n,\ell}^{(\rho)}$, then $\alpha\ge n^{\Omega(1/\ell)}$.
\end{lemma}
\begin{proof}
Consider an integral solution $T$ with out-degree $k_u/\alpha$ at every vertex $u$. In the whole proof, we fix a vertex $v\in L_{1/\epsilon}$, and we define by $T_v$ the subtree of $T$ rooted at a vertex $v$. By construction, if $v\in L_{1/\epsilon}$, then the subtree $T_v$ needs to contain at least $(\alpha)^{-2/\epsilon}\cdot  {(1-\rho)m\choose \rho m}$ sinks. Similarly, for any vertex $u\in L_{2/\epsilon}$ contained in $T$, the subtree $T_u$ must contain $(\alpha)^{-1/\epsilon}\cdot {2\rho m\choose \rho m}$ sinks. Otherwise, the solution $T$ cannot be $\alpha$-approximate.

We remind the reader that here $S_v$ refers to the set corresponding to vertex $v$. We fix some parameter $\theta=\rho/3$, so that $\rho^2<\theta<\rho/2$ (which is true for $\rho$ small enough). Let us partition the set of sinks $L_3$ into two sets. $T_1^{(v)}$ is the set of sinks whose corresponding set $S'$ satisfies $|S'\cap S_v|\ge \theta m$, and $T_2:=L_{3/\epsilon}\setminus T_1^{(v)}$ is the rest of the sinks. 

Let us do some basic counting. We have that 
\begin{align*}
    |T_1^{(v)}| &= \{S'\subseteq \mathcal U\mid |S'|=\rho m\text{ and }|S'\cap S_v|\ge \theta m\}\\
    &=\sum_{j=\theta m}^{\rho m} {|S_v|\choose j}{|\mathcal U\setminus S_v|\choose \rho m-j}\\
    &=\sum_{j=\theta m}^{\rho m} {\rho m\choose j}{m-\rho m\choose \rho m-j} \le m^{O(1)}{\rho m\choose \theta m}{m-\rho m\choose \rho m-\theta m}\ ,
\end{align*} where the inequality uses Lemma \ref{lem:binomial_function} with the fact that $\theta > \rho^2$.

Hence 
\begin{align}
\label{eq:ruling_out_integral_sol1}
    \log_2(|T_1^{(v)}|)/m &\le  \rho h\left(\frac{\theta}{\rho}\right) + (1-\rho)h\left(\frac{\rho-\theta}{1-\rho}\right)+o(1)\\
    &= \rho h(1/3)+(1-\rho)h\left((2/3) \cdot \frac{\rho}{1-\rho}\right)+o(1)\ .
\end{align}

Second, for any vertex $u\in L_{2/\epsilon}$ which is connected to $v$ (i.e. there exists a directed path from $v$ to $u$), we count how many sinks in $T_2$ $u$ can be connected to. Let us denote by $|T^{(u)}_2|$ this number. We recall that $u$ corresponds to a set $S_u$ of size $2\rho m$, and if $u$ is connected to $v$, we must have $S_v\subseteq S_u$. 
\begin{align*}
    |T^{(u)}_2|&=\{S\subseteq \mathcal U\mid |S|=\rho m\text{ and }S\subseteq S_u \text{ and } |S\cap S_v|< \theta m\} \\
    &=\sum_{j=0}^{\theta m} {|S_v|\choose j}{|S_u\setminus S_v| \choose \rho m-j}\\
    &= \sum_{j=0}^{\theta m} {\rho m\choose j}{\rho m \choose \rho m-j}\le m^{O(1)}{\rho m\choose \theta m}{\rho m\choose \rho m-\theta m}=m^{O(1)}{\rho m\choose \theta m}^2\ ,
\end{align*} where the inequality uses the fact that $\theta <\rho / 2 $ with Lemma \ref{lem:binomial_function}.
Hence, for any $u$ connected to $v$, we have
\begin{align}
\label{eq:ruling_out_integral_sol2}
    \log_2(|T^{(u)}_2|)/m &= 2\rho h\left(\frac{\theta}{\rho}\right)+o(1)\\
    &= 2\rho h\left(1/3\right)+o(1)\ .
\end{align}
Note that in the integral solution $T$, there must be some vertex $u\in T_v\cap L_{2/\epsilon}$ which is connected to less than 
\begin{equation*}
    \frac{|T_1^{(v)}|}{(\alpha)^{-1/\epsilon}{(1-\rho)m\choose \rho m}{2\rho m\choose \rho m}^{-1}}
\end{equation*} sinks in $|T_1^{(v)}|$, because there are $(\alpha)^{-1/\epsilon}{(1-\rho)m\choose \rho m}{2\rho m\choose \rho m}^{-1}$ vertices in $T_v\cap L_{2/\epsilon}$ which have to share the set of sinks $T_1^{(v)}$. This vertex $u$ can therefore be connected to at most 
\begin{equation*}
    \frac{|T_1^{(v)}|}{(\alpha)^{-1/\epsilon}{(1-\rho)m\choose \rho m}{2\rho m\choose \rho m}^{-1}}+ |T^{(u)}_2|
\end{equation*} sinks in total in $T$. Therefore, for the solution $T$ to be $\alpha$-approximate, it must be that
\begin{align*}
    &\frac{|T_1^{(v)}|}{(\alpha)^{-1/\epsilon}{(1-\rho)m\choose \rho m}{2\rho m\choose \rho m}^{-1}}+ |T^{(u)}_2| \ge (\alpha)^{-1/\epsilon} \cdot {2\rho m\choose \rho m}\\
    &\implies  \frac{|T_1^{(v)}|}{(\alpha)^{-2/\epsilon} {(1-\rho)m\choose \rho m}}+ \frac{|T^{(u)}_2|}{(\alpha)^{-1/\epsilon} \cdot {2\rho m\choose \rho m}}\ge 1\\
    &\implies \max\left(\frac{|T_1^{(v)}|}{\alpha^{-2/\epsilon} {(1-\rho)m\choose \rho m}}, \frac{|T^{(u)}_2|}{\alpha^{-1/\epsilon} {2\rho m\choose \rho m}}\right)\ge 1/2\ .
\end{align*}
Hence, it must be that either

\begin{align}
    \log_2(|T^{(u)}_2|)/m &\ge -\log_2(\alpha)/(\epsilon m)+ \log_2{2\rho m\choose \rho m}/m -1/m\\
    &=-\log_2(\alpha)/(\epsilon m)+2\rho h(1/2)+O(\log(m)/m)\\
 &=-\log_2(\alpha)/(\epsilon m)+2\rho +O(\log(m)/m)\ ,
\label{eq:ruling_out_integral_sol3}
\end{align} or  
\begin{align}
    \log_2(|T_1^{(v)}|)/m &\ge -2\log_2(\alpha)/(\epsilon m)+\log_2{(1-\rho) m\choose \rho m}/m-1/m\\ 
    &=-2\log_2(\alpha)/(\epsilon m)+(1-\rho)h\left(\frac{\rho}{1-\rho}\right)+O(\log(m)/m)\ .
    \label{eq:ruling_out_integral_sol4}
\end{align}
We remark that $h(1/3)<1$ hence Equation \eqref{eq:ruling_out_integral_sol2} contradicts Equation \eqref{eq:ruling_out_integral_sol3} unless we have that $\log_2(\alpha)\ge \Omega(\epsilon m)$.

Finally, we claim that the function 
\begin{equation*}
    g:\rho \mapsto (1-\rho)h\left(\frac{\rho}{1-\rho}\right)-\rho h(1/3)-(1-\rho)h\left((2/3) \cdot \frac{\rho}{1-\rho}\right)
\end{equation*}
satisfies that $g(\rho)>0$ for $\rho<\xi$ (with $\xi>0$ a small enough constant). This implies that if $\rho<\xi$, we obtain that Equation \eqref{eq:ruling_out_integral_sol1} contradicts Equation  \eqref{eq:ruling_out_integral_sol4}, unless we have that $\log_2(\alpha)\le -\Omega(\epsilon m)$.

Therefore, we reach a contradiction in both cases, unless $\log_2(\alpha)\ge \Omega(\epsilon m)$, which implies that $\alpha\ge 2^{\Omega(\epsilon m)}=n^{\Omega(\epsilon)}=n^{\Omega(1/\ell)}$.

Let us prove the last claim. We use the standard expansion of the entropy function $h(x)= x\frac{(1-\log(x))}{\log(2)}+O(x^2)$ at $x=0$ ($\log$ is the natural logarithm) to obtain that 
\begin{align*}
    g(\rho)&=(1-\rho)\left[\rho\frac{1-\log(\rho/(1-\rho))}{(1-\rho)\log (2)}- (2/3)\rho\frac{1-\log((2/3)\rho/(1-\rho))}{(1-\rho)\log (2)}\right]-\rho h(1/3)+O(\rho^2)\\
    &=\rho\frac{1-\log(\rho/(1-\rho))}{\log (2)}- (2/3)\rho\frac{1-\log((2/3)\rho/(1-\rho))}{\log (2)}-\rho h(1/3)+O(\rho^2)\\
    &= \frac{\rho \log(1/\rho)}{3\log (2)}+O(\rho)\ .
\end{align*}
Clearly, for $\rho>0$ small enough, this expression is strictly positive which concludes the proof.
\end{proof}
\section{The path hierarchy and locally good solutions}
\label{sec:lower_bound_many_rounds}
The goal of this section is to prove Theorem \ref{thm:sherali_many_round} and the last property of Theorem \ref{thm:locality_gap} (the other properties being proven in Section \ref{sec:construction}), in order to give evidence towards our main conjecture. For this we will show that $\Omega(\ell)$ levels of the path hierarchy remain feasible on our instance $G_{n,\ell}^{(\rho)}$, with even some added properties which allow to implement the round-and-condition algorithm in a locally consistent way.

\subsection{Feasible assignment LP solution}
We start by defining a feasible solution to the assignment LP on our instance $G_{n,\ell}^{(\rho)}$. Recall that $\gamma_i=k_i/\delta^+_i$ is the ratio between required out-degrees and out-degree in graph for the vertices in layer $i$. We also recall that $\delta_i^-$ is the in-degree in the graph of vertices in $L_i$. For any edge $e\in L_i$ (i.e. its endpoint is in $L_i$), we set 
\begin{equation}
\label{eq:LP_many_layers}
    x_e := \gamma_{i-1} \cdot \prod_{j=1}^{i-1} (\gamma_{j-1} \delta_j^-)\ ,
\end{equation}
with the convention that $\prod_{j=1}^{0} (\gamma_{j-1} \delta_j^-)=1$.

The covering constraints are easy to check, indeed for any vertex $v\in L_i$, we have that 
\begin{equation*}
    \sum_{e\in \delta^+(v)}x_e = (\gamma_{i} \delta_{i}^+) \cdot \prod_{j=1}^{i} (\gamma_{j-1} \delta_j^-)\ ,
\end{equation*}
and 
\begin{equation*}
    \sum_{e\in \delta^-(v)}x_e =  (\gamma_{i-1}\delta^-_i)\cdot  \prod_{j=1}^{i-1} (\gamma_{j-1} \delta_j^-)=\prod_{j=1}^{i} (\gamma_{j-1} \delta_j^-)\ .
\end{equation*}

The packing constraints are slightly more tricky. We need to compute the quantity
\begin{equation*}
    \sum_{e\in \delta^-(v)}x_e = \prod_{j=1}^{i} (\gamma_{j-1} \delta_j^-)\ ,
\end{equation*}
for any $v\in L_i$. If $i\le 1/\epsilon$, we have 
\begin{equation*}
    \prod_{j=1}^{i} (\gamma_{j-1} \delta_j^-) =\prod_{j=1}^{i}\left[\frac{(\epsilon \rho m)!}{\left((\rho m)!\cdot {(1-\rho) m\choose \rho m}\right)^{\epsilon}}\cdot {j\epsilon \rho m\choose \epsilon\rho m}\right]=\frac{(i\epsilon \rho m)!}{(\rho m)!^{i\epsilon}{{(1-\rho) m\choose \rho m}}^{i\epsilon}}\ .
\end{equation*}
Let us define $x:=i\epsilon \in [0,1]$. We then have 
\begin{equation*}
    \prod_{j=1}^{i} (\gamma_{j-1} \delta_j^-) \le \frac{(x \rho m)!}{(\rho m)!^{x}} \le 1\ .
\end{equation*}

If $1/\epsilon<i\le 2/\epsilon$, then 
\begin{align*}
    \prod_{j=1}^{i} (\gamma_{j-1} \delta_j^-) &= \prod_{j=1}^{1/\epsilon} (\gamma_{j-1} \delta_j^-) \cdot \prod_{j=1/\epsilon+1}^{i} (\gamma_{j-1} \delta_j^-)\\
    &\le {(1-\rho) m\choose \rho m}^{-1}\cdot {2\rho m\choose \rho m}^{1-i\epsilon}\cdot \prod_{j=1/\epsilon+1}^{i}\left[\frac{(\epsilon \rho m)!}{\left((\rho m)!\right)^{\epsilon}}\cdot {j\epsilon \rho m\choose \epsilon\rho m}\right]\\
    &\le {(1-\rho) m\choose \rho m}^{-1}\cdot {2\rho m\choose \rho m}^{1-i\epsilon}\cdot \frac{(i\epsilon \rho m)!}{((\rho m)!)^{i\epsilon}}\ .
\end{align*}
Hence, 
\begin{equation*}
    \log_2\left(\prod_{j=1}^{i} (\gamma_{j-1} \delta_j^-)\right)/m \le -(1-\rho)h\left(\frac{\rho}{1-\rho}\right)+(1-x)\cdot (2\rho)+(\rho x)\cdot \log_2(x)+O(\log m/m)\ .
\end{equation*}
Clearly, this is a convex function of $x$, hence we only need to check the endpoints to obtain the maximum on the whole region $x\in [1,2]$. One can check that at both endpoints the function is strictly negative, hence $\prod_{j=1}^{i} (\gamma_{j-1} \delta_j^-)\le 1$.

Finally, for $i>2/\epsilon$, we have 
\begin{align*}
    \prod_{j=1}^{i} (\gamma_{j-1} \delta_j^-) &= {(1-\rho) m\choose \rho m}^{-1}\cdot \prod_{j=2/\epsilon+1}^{i}\left[\frac{(\epsilon \rho m)!}{\left((\rho m)!\right)^{\epsilon}}\cdot {m-4\rho m+j\epsilon \rho m\choose \epsilon\rho m}\right]\\
    &= {(1-\rho) m\choose \rho m}^{-1}\cdot \frac{(m-4\rho m+i\epsilon \rho m)!}{((\rho m)!)^{i\epsilon-2}\cdot (m-2\rho m)!}\ .
\end{align*}
Clearly, the last expression in increasing in $i$ (recall that $\rho$ is a small constant, which implies that $m-4\rho m>\rho m$). One can conclude by checking the value at $i=3/\epsilon$ to obtain a maximum value of ${(1-\rho) m\choose \rho m}^{-1}\cdot \frac{(m-\rho m)!}{((\rho m)!)\cdot (m-2\rho m)!}=1$.

Hence we can conclude that our choice of $x_e$ variables satisfy the assignment LP.

\subsection{Feasibility of the path hierarchy}

We recall here the path hierarchy. Recall that we keep only variables and constraints for paths of length at most $t+1$, where $t$ is the number of rounds. Recall that $e_0$ is a dummy edge ending at the source vertex (and by convention we assume that $x_{e_0}=1$ where $x$ is the assignment LP solution).
\begin{align}
    \sum_{q\in C(p)} y(q) &= k_p \cdot y(p) &\forall p\in P \label{eqLP:flowLPdemand2}\\
     \sum_{q\in I(v)\cap D(p)} y(q) &\leq y(p) &\forall p\in P,v\in V \label{eqLP:flowLPcapacity2}\\
    \sum_{e\in \delta^-(v)}y(\{e\}) &\le 1 & \forall v\in V \label{eqLP:flowLPpacking_unlifted2}\\
     \sum_{e\in \delta^+(v)}y(\{e\}) &\ge k_v\cdot \sum_{e\in \delta^-(v)}y(\{e\}) & \forall v\in V \label{eqLP:flowLPcovering_unlifted2}\\
    y(q)&\le y(p) & \forall p,q\in P,\  p\subseteq q \label{eqLP:consistency2}\\
    y(\{e_0\}) &= 1 & \label{eqLP:flowLPdemandroot2}\\
    0\le y&\le 1 
\end{align}
To get intuition on how we set the variables, one can imagine that we iterate the shadow distribution $t$ times as explained in introduction. Formally, we imagine that we sample a first set of shadow edges $S^{(1)}$. Each edge $e\in S^{(1)}$ selects a subtree which contains each edge $e'$ independently with probability $x_{e'}^{(e)}$. This gives us a second set of shadows $S^{(2)}$. Now the set $S^{(2)}$ triggers a set $S^{(3)}$ in the same way, and so on until reaching the set $S^{(t+1)}$ which will be our final set.

We do not know how to analyze this procedure precisely, however the probability that a path $p=(e_1,e_2,\ldots ,e_{m})$ of length at most $t+1$ appears in this process seems to be dominated by the probability of the event $\mathcal E$ that $e_1\in S_1$, then $e_2$ is selected in the subtree of $e_1$, then $e_3$ is selected in the subtree of $e_2$, and so on.

One can notice that the most reasonable way to define subtree solutions\footnote{It is possible to prove the existence of subtree solutions for our lifted instances of any depth, however, for better readability, we did not include this result which is not needed in our proof.} $x^{(e)}$ for each edge $e$ in a similar way as in Section \ref{sec:sherali1round} would imply that the probability of this event is given by 
\begin{equation*}
    \mathbb P[\mathcal E] = x_{e_1}\cdot \prod_{j=2}^{m}x_{e_j}^{(e_{j-1})} = x_{e_1}\cdot \prod_{j=2}^{m}\gamma_{e_{j}}\ ,
\end{equation*}
where $\gamma_{e}=\gamma_{i-1}$ for any edge $e\in L_i$. This is exactly how we set our variables. For a directed path $p=(e_1,e_2,\ldots ,e_{m})$ of length at most $t+1$ where $e_1\in L_i$, $e_2\in L_{i+1}$, etc, we set 
\begin{equation}
    \label{eqn:path_solution}
    y(p):=x_{e_1}\cdot \prod_{j=i}^{i+m-1}\gamma_j\ ,
\end{equation}
with the convention that the empty product is equal to 1.

\begin{lemma}
    \label{lem:path_hierarchy_solution}
    There exists some absolute constant $\xi>0$, such that for $t=\xi\cdot \ell$, the solution $(y(p))_{|p|\le t}$ defined in Equation \eqref{eqn:path_solution} is feasible for the $(t-1)^{th}$ level of the path hierarchy on instance $G_{n,\ell}^{(\rho)}$.
\end{lemma}
To prove this lemma, most of the constraints are trivial to check except for the lifted packing constraint (Constraint \eqref{eqLP:flowLPcapacity2}), which requires some involved calculations. For clarity, we defer the proof of the following helper lemma to Appendix \ref{app:number_paths}. Then we prove Lemma \ref{lem:path_hierarchy_solution}.

\begin{lemma}
    \label{lem:helper_lemma_path_hierarchy}
    There exists some absolute constant $\xi>0$, such that for $t=\xi\cdot \ell$, and any vertices $v\in L_i$, and $u\in L_{i\le j\le i+t}$, the number of directed paths from $v$ to $u$ in $G_{n,\ell}^{(\rho)}$ is at most 
    \begin{equation*}
        \frac{1}{\prod_{k=i}^{j-1}\gamma_k}\ .
    \end{equation*}
\end{lemma}

\begin{proof}[Proof of Lemma \ref{lem:path_hierarchy_solution}]
    Most of the constraints are trivial to check. First, note that for a single edge (i.e. a path of length 1), we have $y(\{e\})=x_e$ so Constraints \eqref{eqLP:flowLPcovering_unlifted2} and \eqref{eqLP:flowLPpacking_unlifted2} are clearly satisfied. For the consistency constraint (Constraint \eqref{eqLP:consistency2}) for some $p=(e_1,e_2,\ldots ,e_m)\subseteq q=(e'_1,\ldots , e'_{m'},p,e''_1,\ldots , e''_{m''})$, where we denote by $i'$ the layer to which $e'_1$ belongs, and $i$ the layer to which $e_1$ belongs, we have 
\begin{align*}
    \frac{y(q)}{y(p)} &\le \frac{x_{e'_1}\cdot (\prod_{j=2}^{m'}\gamma_{e'_j}) \cdot \gamma_{e_1}}{x_{e_1}} = \frac{(\prod_{j=1}^{i'-1} \gamma_{j-1} \delta_j^-) (\prod_{j=i'}^{i}\gamma_{j-1})}{\gamma_{i-1}\prod_{j=1}^{i-1} (\gamma_{j-1} \delta_j^-)}\\
    &=\frac{1}{\prod_{j=i'}^{i-1}\delta_j^-}\le 1\ .
\end{align*}
The lifted covering constraints (Constraint \eqref{eqLP:flowLPdemand2}) are also easy to check. Denote by $v\in L_i$ the endpoint of $p$. Then, we have

\begin{equation*}
     \sum_{q\in C(p)} y(q) = y(p)\cdot \sum_{e\in \delta^+(v)}\gamma_i = (\gamma_i\delta^+_i)\cdot y(p)=k_p\cdot y(p)\ .
\end{equation*}
Finally, Constraint \eqref{eqLP:flowLPcapacity2} is implied by Lemma \ref{lem:helper_lemma_path_hierarchy} in a straightforward way. Indeed, note that for any path $p=(s,\ldots ,u)$ and any vertex $v$, the variable corresponding to a path $q\in D(p)\cap I(v)$ is equal to the variable $y(p)$, multiplied by one $\gamma_k$ for each additional laye $k$ from $u\in L_i$ to $v\in L_j$ (hence an additional multiplicative factor of $\prod_{k=i}^{j-1}\gamma_k$). But by Lemma \ref{lem:helper_lemma_path_hierarchy}, the number of such paths is less than 
\begin{equation*}
    \frac{1}{\prod_{k=i}^{j-1}\gamma_k}\ ,
\end{equation*}
which cancels out exactly with the multiplicative factor. Note that we used the assumption that our paths have length at most $t$ to be able to apply Lemma \ref{lem:helper_lemma_path_hierarchy}.
\end{proof}

\subsection{Locally good solutions}

Here, we describe the algorithm that we use to obtain our locally good solutions on our instance. The arguments here are rather standard, see e.g. \cite{bamas2023better,chakrabarty2009allocating, bateni2009maxmin}. In the following, it is useful to think of the output as a set $P'$ of directed paths, where each path $P'$ starts at the source.

The algorithm proceeds layer by layer.
\begin{enumerate}
    \item Sample each edge in $e\in \delta^+(s)$ independently with probability $\gamma_0$. For each edge $e$ sampled, add the path $\{e\}$ to $P'$.
    \item For $j$ from $1$ to $\ell$, for each path $p=(s,\ldots, v)$ from the source to $v\in L_j$ which was selected in $P'$, select each edge $e\in \delta^+(v)$ independently with probability equal to $\gamma_i$ and add the path $p\circ \{e\}$ to $P'$ (where $\circ$ is the concatenation operator).
\end{enumerate}

We claim the following lemma regarding this algorithm. Let $P'$ the set of paths sampled by the algorithm. Then we have the following. 
\begin{lemma}\label{locally_good_lemma}
With high probability, the set $P'$ is a $(\xi \ell)$-locally good solutions, i.e. the following constraints hold
\begin{align}
    \sum_{q\in C(p)} y(q) &\ge  \Omega(k_p) \cdot y(p) &\forall p\in P' \label{eqLP:flowLPdemand3}\\
     \sum_{q\in I(v)\cap D(p)} y(q) &\leq (\log n)^{O(1)} \cdot y(p) &\forall p\in P,v\in D(p,\xi\ell-|p|) \label{eqLP:flowLPcapacity3}
\end{align}
where $D(p,t)$ is the set of vertices for which there exists a directed path of length at most $t$ from the endpoint of $p$ to $v$, and $k_p=k_v$ where $v$ is the last vertex of $p$ (with the convention that $k_v=0$ if $v$ is a sink).
\end{lemma}
\begin{proof}
    First for Constraint \eqref{eqLP:flowLPdemand3}, we note that clearly whenever a path $p=(s,\ldots, v)$ with $v\in L_i$ is selected in our solution, the algorithm selects in expectation $(\gamma_i\delta^+_i)$ paths $q\in C(p)$ in expectation, and in an independent manner. Hence in expectation, the path $p$ receives at least $k_p$ children. Now recall that by Lemma \ref{lem:large_required_degrees}, $k_p=\Omega(\polylog)$ for all $p$, hence we can apply Chernoff bound, and with high probability, each path $p$ receives say $k_p/2$ children.

    For Constraint \eqref{eqLP:flowLPcapacity3}, this is more tricky. We proceed by induction layer by layer. We fix a path $p=(s,\ldots, u)$ with $u\in L_i$ and some $v\in L_{j\ge i}$. Clearly, conditioned on selecting $p$, each path $q\in D(p)\cap I(v)$ is selected with probability 
    \begin{equation*}
        \frac{1}{\prod_{k=i}^{j-1}\gamma_k}\ .
    \end{equation*}
    By Lemma \ref{lem:helper_lemma_path_hierarchy}, this implies that in expectation the constraint is satisfied. However, the paths do not appear independently of each other so we need to be more careful. By induction we prove the following. 
    
    Let $X$ be the number of paths $q\in D(p)\cap I(v)$ selected in the solution. Then for any $t\le j-i$, we have with high probability that 
    \begin{equation*}
        \mathbb E_{[t+1,\ldots ,j-i]}[X\mid \mathcal P_t]\le \log^{10} (n)\cdot \left(1+\frac{1}{\log n}\right)^t\ ,
    \end{equation*}
    where $\mathcal P_t$ is the outcome of the sampling algorithm at layer $i+t$ and the randomness in the expectation is the outcome of the rouding in layers $i+1,\ldots t$. We argued that this was true for $t=0$. Assume this is true for some $t$, let us now prove it for $t+1$. For each path $q'$ selected at layer $i+t$, we will sample children paths in layer $i+t+1$ with the adequate probability so clearly we have that 
    \begin{equation*}
        \mathbb E_{t+1}[\mathbb E_{[t+2,\ldots, j-i]}[X\mid \mathcal P_{t+1}]] = \mathbb E_{[t+1,\ldots ,j-i]}[X\mid \mathcal P_t] = \log^{10} (n)\cdot \left(1+\frac{1}{\log n}\right)^t\ ,
    \end{equation*}
    where $\mathbb E_{t+1}$ means that we take the expectation on the random choices made in layer $i+t+1$. Second, by Constraint \eqref{eqLP:flowLPcapacity2} and Lemma \ref{lem:path_hierarchy_solution}, we have that, for each path $q'$ selected in layer $i+t+1$, the expected number of paths $q''\in D(q')\cap I(v)$ that will be selected in subsequent layers is at most $1$. Hence the random variable 
    $\mathbb E_{[t+2,\ldots, j-i]}[X\mid \mathcal P_{t+1}]$ can be written as a sum of independent random variables of absolute value at most $1$, and of expectation at most $\log^{10} (n)\cdot \left(1+\frac{1}{\log n}\right)^t$. By Chernoff bounds, with high probability this random variable is no more than $\log^{10} (n)\cdot \left(1+\frac{1}{\log n}\right)^{t+1}$, which proves the inductive step.

    We can conclude that, with high probability, the number of paths $q\in D(p)\cap I(v)$ selected is at most $\log^{10}(n)\cdot (1+1/\log n)^{\ell}\le \log^{11}(n)$.
\end{proof}

\section{Future directions}
We conclude by listing a few interesting directions which, in our opinion, are of interest. We start by our main conjecture on the Sherali-Adams hierarchy, but some questions have a more distant relation to our work.
\begin{enumerate}
    \item Does our instance of depth $\ell$ survive $\Omega(\ell)$ rounds of the Sherali-Adams hierarchy? Another interesting direction would be to simplify our construction.
    \item Does our instance of depth $\ell$ survive $\Omega(\ell)$ rounds of stronger hierarchies such as the Lasserre hierarchy? A good starting point could be to understand if our instance of depth $3$ survives $1$ round of the Lasserre hierarchy.
    \item Is there a formal relationship between the Directed Steiner Tree problem and the MMDA problem? It has been known for a long time that the DST problem contains the set cover problem as a special case, and recently noticed by the authors of \cite{bamas2023better} that MMDA contains the max-$k$-cover problem as a special case. It is also well-known that there is a reduction from set cover to max-$k$-cover, which shows that a constant-factor approximation for max-$k$-cover implies an $O(\log n)$-approximation for set cover. It is tempting to believe that there is also such a reduction between the MMDA problem and the DST problem. One can also ask a similar question regarding other related problems such as the densest-$k$-subgraph problem \cite{bhaskara2010detecting}, or orienteering problems \cite{chekuri2005recursive}.
    \item Is there an Exponential-Time-Hypothesis-based hardness for the MMDA problem? The reason to ask this question is that this kind of labeling scheme based on a set system has been used in the past for ETH-based reductions for related problems such as Densest-$k$-subgraph (see e.g. \cite{manurangsi2017birthday,manurangsi2017almost}). This is a vague connection, but the intuitive way of transforming our instance into an ETH-based hardness would create an exponential size reduction, which would explain the separation that seems to exist between polynomial-time and quasi-polynomial-time algorithms for the Santa Claus problem.
    \item Is there a constant-factor approximation in quasi-polynomial time for the MMDA problem? The closest that we have until now is a polyloglog$(n)$-approximation in time $n^{O(\log n)}$ for the MMDA problem, and only in the case where $k_u=k$ for all $u$ \cite{bamas2023better}. Even in layered graphs we do not know any better than this.
\end{enumerate}



\ifanonymous{\section*{Acknowledgments}
I am very thankful to Lars Rohwedder for many discussions on this problem, as well as for allowing me to include Theorem \ref{thm:sherali_restricted_lowerbound} and Theorem \ref{thm:sherali_restricted_upperbound}.}{}
\appendix
\section{Restricted Assignment}
\label{sec:restricted}

Restricted assignment is a special case of the
Santa Claus problem where each resource $j$ has a fixed
value $v_j$ independent of the players. On the other hand,
each player $i$ has a set
of eligible resources $R(i)\subseteq R$ and only these
can be assigned to $i$.
This is equivalent to requiring that $v_{ij}\in\{0, v_j\}$
for all $i,j$.
The assignment LP for restricted assignment is defined
as
\begin{align*}
    \sum_{j\in R(i)} v_j x_{ij} &\ge T &\forall i\in P \\
    \sum_{i\in P : j\in R(i)} x_{ij} &\le 1 &\forall j\in R \\
    x_{ij} &\in [0, 1] &\forall i\in P, j\in R(i)
\end{align*}
The restricted assignment case has been studied extensively
and is fairly well understood. Regarding Sherali-Adams,
it is interesting in two ways:
we show that even after a linear number of rounds of Sherali-Adams the integrality gap remains unbounded.
Fortunately, there is a simple way to resolve this issue.
We consider canonical instances that
we will introduce later, which arise from an approximation 
preserving simplification to the structure of the problem,
and have been used in previous works as well.
We show that on canonical instances already one round
of Sherali-Adams yields a constant integrality gap.
This is by proving that it is at least as strong as another
known linear programming relaxation.

\subsection{Lower bound for linearly many rounds}
Let $0 < \epsilon \le 1/3$ and $k\in\mathbb N$ with $\epsilon k \in \mathbb N$.
Consider the following instance of restricted assignment.
There are $k+1$ players indexed by $1,2,\dotsc,k+1$
and $k+1$ small resources $s_1,s_2,\dotsc,s_{k+1}$ with
$v_{s_j} = 3\epsilon$ for all $j$.
Further, there are $k$ big resources $b_1,b_2,\dotsc,b_k$
with $v_{b_j} = 1$ for each $j$.
Each player has access to one small resource and all
big resources, that is, $R(i) = \{s_i, b_1,b_2,\dotsc,b_k\}$
for each $i$.

The integral optimum of this instance is clearly $3\epsilon$,
since not every player can receive a big resource and
the one that does not can only get a value of $3\epsilon$.
We will now show that there is a solution for SA($r$) with value $1$ even for $r = \epsilon k$.

We will use Theorem \ref{thm:sherali_distributions} to argue about the feasibility of the Sherali-Adams hierarchy. To this end, consider the following probability distribution.
We select a matching of the $k$ big resources to $k$ of the players uniformly at random.
The small resources are deterministically assigned to their corresponding players.

It is clear that this is a distribution over valid assignments of resources.
Hence, we only need to verify that for each player the expected value
of resources assigned to it is at least $1$, 
even after conditioning on $\epsilon k$
variables being either $0$ or $1$.
We can assume without loss of generality that all these variables
are conditioned to be $1$:
in every realization every resource is assigned
to exactly one player.
Thus, if we condition on a resource not to be assigned to a
particular player, then we get a probability distribution
over solutions where this resource is assigned to another
player. Conditioning on this will not increase the number
of variables we condition on, but refine the distribution further.

This means that there are $\epsilon k$ players that receive a specific big resource. The remaining $(1 - \epsilon)k$ big resources are still uniformly
assigned to the remaining $1 + (1 - \epsilon)k$ players.
Let $i$ be a player.
If $i$ is one of the players that are guaranteed to receive a big
resource, then clearly the expected value for $i$ is at least $1$.
Otherwise, there is a probability of $(1 - \epsilon)k / (1 + (1 - \epsilon) k)$
that $i$ receives a big resource. Additionally, the player is
guaranteed to receive a small job of value $3\epsilon$.
Thus, in expectation the value received is at least
\begin{equation*}
    \frac{(1 - \epsilon)k}{1 + (1 - \epsilon)k} + 3 \epsilon 
    \ge \frac{(1 - \epsilon)k + 3\epsilon (1 - \epsilon)k}{1 + (1 - \epsilon)k}
    \ge 1 \ .
\end{equation*}
\subsection{Canonical restricted assignment}
Suppose we want to solve the following variant, which is 
equivalent to an $\alpha$-approximation algorithm:
given some $T \ge 0$ either determine that $\OPT < T$ or find a
solution of value $\alpha T$.
Let $B\subseteq R$ be the resources $j$ with $v_j \ge \alpha T$
and let $S = R\setminus B$.
Via a simple transformation we can assume that each player either
only has access to big resources, or to at most one big resource and otherwise
to small resources. Further, $v_j = T$ for all $j\in B$.
We call such an instance an $\alpha$-canonical instance.
Similar simplifications are standard for the problem,
see e.g.~\cite{chakrabarty2009allocating,bansal2006santa}.

For the transformation we introduce two players and a coupling resource of value $T$ for each original player. We give both players access to the coupling resource and we connect one
player to the small resources the original player had access to and one to the
big resources. Further, we increase the value of each big resource to $T$.
The transformation satisfies that 
if there is a solution of value $T$ for the original instance,
there is also one solution for the canonical instance.
Further, if there is a solution of value $\alpha T$ for the
canonical instance, there is one for the original instance.
In that sense, the transformation is approximation preserving.

We will now show that for $\alpha = O(1)$, an $\alpha$-normalized
instance has an integrality gap of $\alpha$ already after a single round of Sherali-Adams.
Our proof is by showing that SA(1) is at least as strong
as an LP described by Davies et al. \cite{davies2020tale}, who prove that the following linear program has an
integrality gap of at most $4$:
\begin{align*}
    \sum_{j\in R(i)} v_j x_{ij} &\ge T &\forall i\in P \\
    \sum_{i\in P : j\in R(i)} x_{ij} &\le 1 &\forall j\in R \\
    x_{is} &\le 1 - \sum_{j\in B \cap R(i)} x_{ij} &\forall i\in P\ \forall s\in S \cap R(i) \\
    x_{ij} &\ge 0 &\forall i\in P, j\in R(i)
\end{align*}
Consider now a canonical instance with a solution $y\in \mathrm{SA}(1)$.
We will define a solution $x$ for the above linear program from this.
Let $i$ be a player. We set 
$x_{ij} = y_{ij}$ for all $i\in P$ and big resources $j\in R(i) \cap B$.
Let $i$ be a player that has access to only a single big resource $b$, then we set
$x_{ij} = (1 - y_{i b}) * y_{ij} = y_{ij} - y_{\{ij, ib\}}$
for all $j\in R(i) \cap S$.
To verify that this is a feasible solution to the linear
program above, we start with the first constraints.
Let $i\in P$. If $i$ has only access to big resources,
then $\sum_{j\in R(i)} v_j x_{ij} = \sum_{j\in R(i)} v_j y_{ij} \ge T$.
If $i$ has access to small resources and a single big resource $b$, then use that
\begin{align*}
\sum_{j\in S \cap R(i)} v_j (y_{ij} - y_{\{ij, i b\}})
&= (1 - y_{i b}) * \sum_{j\in R(i)} v_j y_{ij} \\
&\ge (1 - y_{i b}) * T = T - T \cdot y_{i b} \ .
\end{align*}
Thus,
\begin{equation*}
    \sum_{j\in R(i)} v_j x_{ij} = T \cdot y_{ib} + \sum_{j\in S \cap R(i)} v_j (y_{ij} - y_{\{ij, ib\}}) \ge T \ .
\end{equation*}
If $i$ does not have access to small resources, then it follows immediately that
\begin{equation*}
    \sum_{j\in R(i)} v_j x_{ij} = \sum_{j\in R(i)} v_j y_{ij} \ge T \ .
\end{equation*}
Since $x_{ij} \le y_{ij}$ for all $i,j$ it also holds that $\sum_{i\in P : j\in R(i)} x_{ij} \le \sum_{i\in P : j\in R(i)} y_{ij} \le 1$ for all $j\in R$.
Finally, for some player $i$ that has only access to single big resource $b$ we have $(1 - x_{i b}) * x_{is} \le (1 - x_{i b}) * 1$ for all $s\in S\cap R(i)$. Hence,
\begin{equation*}
    x_{is} = y_{is} - y_{\{is, ib\}} \le 1 - y_{i b} = 1 - x_{i b} = 1 - \sum_{j\in B\cap R(i)} x_{ij} \ .
\end{equation*}
\section{Previous constructions}
\label{app:previous_construction}

A long line of works on the restricted assignment case used a popular relaxation called the \textit{configuration LP}, which has exponential size but is still solvable in polynomial time (by loosing some small $1+\epsilon$ factor). Since the introduction of this relaxation in \cite{bansal2006santa}, it was already known that it has a polynomial integrality gap on general instances, while it is known to have a constant integrality gap for the restricted assignment \cite{asadpour2012santa}. We can construct an MMDA instance of depth $3$ which has a very similar structure to the integrality gap instance of Bansal and Srividenko (see Figure \ref{fig:bansal_instance}). 

To see the correspondence with Santa Claus, consider the following simple transformation: replace the sink by a player, each sink by a resource, and all other vertices by a pair player-resource, where the resource is the \textit{private} resource of the player. Now the valuation functions are given as follows: the source player values each private resource of a player in $L_1$ by a value $1/k$. Each other player values its own private resource by a value of $1$, and the private resources of players in the next layer he is connected to by a value $1/k$. Finally, every player in $L_2$ values the resources in $L_3$ he is connected to by a value $1/k$. All other pairs resource-player have value $0$. See Figure \ref{fig:bansal_instance_Santa} for the description of this instance.

After this transformation, our instance in Figure \ref{fig:bansal_instance} becomes very similar to the gap instance of \cite{bansal2006santa}, and it is easy to see that configuration LP is feasible for the objective value of $1$: the source player will simply take $k$ disjoint configurations of size $k$ with value $1/k$ each, and below this, all the players take their private resource with value $1-1/k$, and the other possible configuration with value $1/k$.

\begin{figure}[h]
    \centering
    \includegraphics{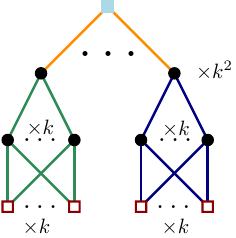}
    \caption{An instance similar to \cite{bansal2006santa} translated as an MMDA instance. The source is connected to $k^2$ vertices in $L_1$, each vertex in $L_1$ is connected to a ``private'' subgraph containing $k$ vertices in $L_2$ and $k$ sinks in $L_3$. In every private subgraph, each vertex in $L_2$ is connected to all sinks in $L_3$. Finally, we set $k_u=k$ for all non-sink vertices.}
    \label{fig:bansal_instance}
\end{figure}

\begin{figure}[h]
    \centering
    \includegraphics{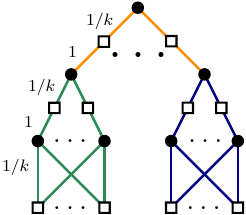}
    \caption{The Santa Claus instance corresponding to the MMDA instance in Figure \ref{fig:bansal_instance}. Circles represent players and squares represent resources. The valuation for edges is given on the left. If there is no edge between a player and a resource, that player values this resource to $0$.}
    \label{fig:bansal_instance_Santa}
\end{figure}

Clearly, this instance has a size polynomial in $k$, and does not contain a feasible solution with out-degree more than $\sqrt{k}$, because every vertex in $L_1$ can reach at most $k$ sinks. In the Santa Claus instance, this means that the optimum integral solution has value at most $1/\sqrt{k}$. An implication of this is that there is no subtree solution rooted at a vertex $v\in L_1$. With this condition in mind, it is clear that $1$ round of Sherali-Adams does not get fooled, since this hierarchy will be able to condition by the event that we take an edge in the first layer with value $1$.

As such, this construction fools the configuration LP by a polynomial factor, but not one round of Sherali-Adams.

\section{The existence of subtree solutions is not enough to fool Sherali-Adams}
\label{app:subtree}
Here, we show that there exists an instance of the MMDA problem which has a polynomial gap with the naive LP, and contains a subtree solution for every edge, yet one round of Sherali-Adams has a gap of at most $O(\polylog)$. For this, consider the simple example given in Figure \ref{fig:gap_subtree}. This instance has $O(k^3)$ edges and $O(k^2)$ vertices. One can check that there exists a feasible assignment LP solution: simply set $x_e=1/k$ on orange edges leaving the source, and $x_e=1$ on all green edges between a vertex in $L_1$ and its private sink. Second, there exists a subtree solution for any edge $e'$: for any edge $e'\in L_1$ ending at vertex $v$, we set $x_e^{(e')}=1$ on all blue edges connecting $v$ to the $k$ public sinks.

Clearly, no integral solution can give more than $\sqrt{k}+1$ out-degree to all its vertices, since each vertex in $L_1$ has only one private sink, and has to share only $k$ public sinks with all other vertices. Hence the gap with the naive LP is $k^{\Omega(1)}=n^{\Omega(1)}$. However, one round of Sherali-Adams is not fooled by this since this instance has only depth $2$.

Also note that if we were to consider the shadow distribution (one round) induced by these subtree solutions, then for any blue edge $e=(u,v)$ between a vertex in $L_1$ and a public sink, we would have that
\begin{equation*}
    s_e = x_{(s,u)}\cdot x_{e}^{(u,v)} = 1/k\ ,
\end{equation*}
but $x_e=0$. Hence the key property that $x_e\le s_e\le O(x_e)$ is lost.

\begin{figure}[h]
    \centering
    \includegraphics[scale=1.2]{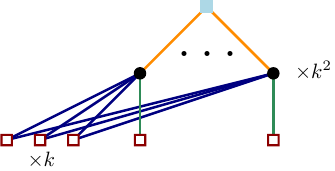}
    \caption{An example of depth $2$ which contains a subtree solution $x^{(e')}$ for any $e'$ and fools the naive LP, but does not fool the Sherali-Adams hierarchy. In this example, all non-sink vertices require out-degree $k$. There are $k^2$ vertices connected to the source, each of them connected to a single private sink, and an additional set of $k$ public sinks which are shared with all other vertices.}
    \label{fig:gap_subtree}
\end{figure}

\section{Proof of Lemma \ref{lem:helper_lemma_path_hierarchy}}
\label{app:number_paths}
Recall that we want to prove the following statement. There exists some absolute constant $\xi>0$, such that for $t=\xi\cdot \ell$, and any vertices $v\in L_i$, and $u\in L_{i\le j\le i+t}$, the number of directed paths from $v$ to $u$ in $G_{n,\ell}^{(\rho)}$ is at most 
    \begin{equation*}
        \frac{1}{\prod_{k=i}^{j-1}\gamma_k}\ .
    \end{equation*}
We proceed by a case analysis. Let $v\in L_i,u\in L_j$, with the property that there exists at least one directed path from $v$ to $u$. 
\paragraph{Case 1: If $u,v\in L_{\le 2/\epsilon}$.} This case is in fact quite easy. It must be that $S_u\subseteq S_v$, and the number of paths is exactly equal to the number of ordering $|S_v|-|S_u|$ elements into $(|S_v|-|S_u|)/(\epsilon \rho m)$ buckets of size $\epsilon \rho m$ (the order inside each bucket is not counted). Hence the number of paths is equal to
\begin{equation*}
    \frac{(|S_v|-|S_u|)!}{((\epsilon \rho m)!)^{(|S_v|-|S_u|)/(\epsilon \rho m)}}=\frac{((j-i)\epsilon \rho m)!}{(\epsilon \rho m)!^{j-i}}\ .
\end{equation*}
Also note that for all $0\le k< 3/\epsilon$
\begin{equation*}
    \frac{1}{\gamma_k}\ge \frac{(\rho m)!^{\epsilon}}{(\epsilon \rho m)!}\ ,
\end{equation*}
hence
\begin{equation*}
    \frac{1}{\prod_{k=i}^{j-1}\gamma_k}\ge \left(\frac{((\rho m)!)^{\epsilon}}{(\epsilon \rho m)!}\right)^{j-i}\ .
\end{equation*}
Hence the number of paths, divided by $\frac{1}{\prod_{k=i}^{j-1}\gamma_k}$ is at most 
\begin{equation*}
    \frac{((j-i)\epsilon \rho m)!}{((\rho m)!)^{(j-i)\epsilon}}\ ,
\end{equation*}
which is less than $1$ if $j-i\le 1/\epsilon=\Omega(\ell)$.

\paragraph{Case 2: If $u,v\in L_{\ge 2/\epsilon}$.} The calculations from the previous case also apply here. The roles of $u$ and $v$ are simply reversed.

\paragraph{Case 3: If $v\in L_{\le 2/\epsilon},u\in L_{\ge 2/\epsilon}$.} This is the tricky case. We assume that $v\in L_{\ge 1/\epsilon}$ for simplicity (note that we only need the lemma to hold for any $|j-i|\ge \xi\ell$ for $\xi>0$ some small constant, so this is wlog). To compute $p_{v,u}$ (the number of paths from $v$ to $u$), we first see that there are exactly 
\begin{equation*}
    {m-|S_u\cup S_v|\choose 2\rho m-|S_u\cup S_v|}
\end{equation*}
vertices $w\in L_{2/\epsilon}$ which are reachable from $v$, and from which we can reach $u$ (here $S_u$ is the set corresponding to vertex $u$, and $S_v$ the set corresponding to $v$). Indeed, there is one such vertex for each set of size $2\rho m$ containing both $S_u$ and $S_v$ as a subset.

Hence, to obtain the value of $p_{v,u}$, it suffices to fix one such vertex $w$, and then count how many ways there are to go from $v$ to $w$, and from $w$ to $u$. Using the calculations from the previous cases, we obtain that 
\begin{equation*}
    p_{v,u} =   {m-|S_u\cup S_v|\choose 2\rho m-|S_u\cup S_v|} \cdot \frac{((j-2/\epsilon)\epsilon \rho m)!((2/\epsilon-i)\epsilon \rho m)!}{(\epsilon \rho m)!^{(j-i)}}\ .
\end{equation*}
For fixed $i$ and $j$, we argue that this expression is decreasing as $|S_u\cup S_v|$ increases. Formally, if we define the function
\begin{equation*}
    \theta(x):={m-x\choose 2\rho m-x}\ ,
\end{equation*}
then 
\begin{equation*}
    \frac{\theta(x+1)}{\theta(x)}=\frac{(2\rho m-x)}{m-x}\le 1\ ,
\end{equation*}
for all $0\le x\le 2\rho m$ (recall that $\rho$ is a small constant).

Hence the worst-case is when $|S_u\cup S_v|$ is as small as possible. We have then 2 subcases.
\paragraph{Subcase 3.1: $|S_u|\ge |S_v|$.} This is the case when $j-2/\epsilon\le 2/\epsilon - i$. In that case, we use that 
\begin{equation*}
    |S_u\cup S_v|\ge |S_u| = 4\rho m-j\epsilon \rho m\ ,
\end{equation*}
and we write 
\begin{align*}
    p_{v,u} &\le    {m-4\rho m+j\epsilon \rho m\choose j\epsilon \rho m-2\rho m} \cdot \frac{((j-2/\epsilon)\epsilon \rho m)!((2/\epsilon-i)\epsilon \rho m)!}{(\epsilon \rho m)!^{(j-i)}}\ .
\end{align*}
Also note that 
\begin{equation*}
    \prod_{k=i}^{j-1}\gamma_k = \left(\prod_{k=i}^{2/\epsilon-1} \frac{(\epsilon \rho m)!}{(\rho m)!^{\epsilon}\cdot {2\rho m\choose \rho m}^{\epsilon}} \right)\cdot \left(\prod_{k=2/\epsilon}^{j-1} \frac{(\epsilon \rho m)!}{(\rho m)!^{\epsilon}}\right)\ . 
\end{equation*}
Therefore, using Stirling's formula, we get
     \begin{align*}
        \frac{p_{v,u}}{\prod_{k=i}^{j-1}1/\gamma_k}&\le  {m-4\rho m+j\epsilon \rho m\choose j\epsilon \rho m-2\rho m} \cdot \frac{((j-2/\epsilon)\epsilon \rho m)!((2/\epsilon-i)\epsilon \rho m)!}{(\epsilon \rho m)!^{(j-i)}} \\
        & \hspace{5cm}\cdot \left(\frac{((\epsilon \rho m)!)^{2/\epsilon-i}}{\left(((\rho m)!){2\rho m\choose \rho m}\right)^{2-i\epsilon}}\right) \cdot \left(\frac{((\epsilon \rho m)!)^{j-2/\epsilon}}{\left(((\rho m)!)\right)^{j\epsilon -2}}\right)\\
        &= {m-4\rho m+j\epsilon \rho m\choose j\epsilon \rho m-2\rho m} \cdot \frac{((j-2/\epsilon)\epsilon \rho m)!((2/\epsilon-i)\epsilon \rho m)!}{((\rho m)!)^{j\epsilon -i\epsilon}{2\rho m\choose \rho m}^{2-i\epsilon}} \\
        &\le  m^{O(1)}\cdot \frac{{m-4\rho m+j\epsilon \rho m\choose j\epsilon \rho m-2\rho m}}{{2\rho m\choose \rho m}^{2-i\epsilon}} \cdot (j\epsilon - 2)^{(j\epsilon -2)\rho m}\cdot (2-i\epsilon)^{(2-i\epsilon)\rho m}
\end{align*}
Let us define $x:=i\epsilon$ and $y:=j\epsilon$. Recall that in the subcase we are working in, we have $1\le x\le 2$, $2\le y\le 3$, and $y\le 4-x$. We also recall that $h(x):=-x\log_2(x)-(1-x)\log_2(1-x)$ is the entropy function of the Bernoulli distribution with parameter $x$. Then we get that 
\begin{align*}
    \log_2 &\left( \frac{p_{v,u}}{\prod_{k=i}^{j-1}1/\gamma_k}\right)/m -O(\log (m)/m)\\
    &\le (1-4\rho +y\rho)h\left(\frac{\rho (y-2)}{1-4\rho +y\rho}\right)-(2-x)2\rho h(1/2)+\rho((y-2)\log_2(y-2)+(2-x)\log_2(2-x))\\
    &= (1-4\rho +y\rho)h\left(\frac{\rho (y-2)}{1-4\rho +y\rho}\right)-(2-x)2\rho+\rho((y-2)\log_2(y-2)+(2-x)\log_2(2-x))\\
    &= (2-x)\rho(\log_2(2-x)-2)+(1-4\rho +y\rho)h\left(\frac{\rho (y-2)}{1-4\rho +y\rho}\right)+\rho(y-2)\log_2(y-2)\ .
\end{align*}
On can compute the partial derivative of the above function in variable $x$ and obtain the following derivative
\begin{equation*}
    \rho \left(2-\frac{1}{\log(2)}\left(1+\log(2-x)\right)\right)\ ,
\end{equation*}
where $\log$ is the natural logarithm. This expression is positive as long as $x>2-\exp\left(2\log(2)-1\right)\approx 0.53$. Hence the above expression is maximized for $x$ as big as possible. Let us further impose the constraint that $x>2-\delta_1$ for some small $\delta_1>0$ which will be fixed later. Then our constraints become that $2-\delta_1 \le x\le 2$, $2\le y\le 4-x <2+\delta_1$. Hence we replace $x=4-y$ (which is the biggest value possible for $x$ given our constraints) and obtain that  

\begin{align*}
     \log_2 &\left( \frac{p_{v,u}}{\prod_{k=i}^{j-1}1/\gamma_k}\right)/m-O(\log(m)/m)\\
    &\le \max_{2\le y\le 2+\delta_1} (y-2)\rho(\log_2(y-2)-2)+(1-4\rho +y\rho)h\left(\frac{\rho (y-2)}{1-4\rho +y\rho}\right)+\rho(y-2)\log_2(y-2)\\
    &= \max_{2\le y\le 2+\delta_1} 2(y-2)\rho(\log_2(y-2)-1)+(1-4\rho +y\rho)h\left(\frac{\rho (y-2)}{1-4\rho +y\rho}\right) =: \max_{2\le y\le 2+\delta_1} f_1(y)\ .
\end{align*}
One can compute $\lim_{y\rightarrow 2^+}f_1(y)=0$ and that
\begin{align*}
    f_1'(y)&=-\frac{\rho}{\log(2)} \left(\log(4)-2+\log\left(\frac{\rho(y-2)}{1+\rho(y-4)}\right)-2\log(y-2)\right)\\
    &= -\frac{\rho}{\log(2)} \left[\log(4)-2+\log\left(\frac{\rho}{(1+\rho(y-4))(y-2)}\right)\right]\ ,
\end{align*}
which is negative for $y$ close to $2$ (we even have $\lim_{y\rightarrow 2^+}f_1'(y)=-\infty$). Hence by continuity, there exists a small $\delta_1>0$ ($\delta_1$ depends only on $\rho$) such that $\max_{2\le y\le 2+\delta_1}f_1(y)=f_1(2)=0$. 

The conclusion of Subcase 3.1 is that as long as $2-\delta_1 \le x\le 2$, $2\le y\le 4-x$, we obtain that 

\begin{equation*}
   \frac{p_{v,u}}{\prod_{k=i}^{j-1}1/\gamma_k} \le 1\ .
\end{equation*}
    
\paragraph{Subcase 3.2: $|S_u|\le |S_v|$.} This is the case when $j-2/\epsilon\ge 2/\epsilon - i$. In that case, we use that 
\begin{equation*}
    |S_u\cup S_v|\ge |S_v| = i\epsilon \rho m\ .
\end{equation*}
Here, the same calculations as in Subcase 3.1 apply, except that we replace the binomial coefficient 
\begin{equation*}
    {m-4\rho m+j\epsilon \rho m\choose j\epsilon \rho m-2\rho m}
\end{equation*} by the binomial coefficient

\begin{equation*}
    {m-i\epsilon \rho m\choose 2\rho m-i\epsilon \rho m}\ .
\end{equation*}
    
We obtain similarly that
\begin{align*}
         \frac{p_{v,u}}{\prod_{k=i}^{j-1}1/\gamma_k}  &\le  m^{O(1)}\cdot \frac{ {m-i\epsilon \rho m\choose 2\rho m-i\epsilon \rho m}}{{2\rho m\choose \rho m}^{2-i\epsilon}} \cdot (j\epsilon - 2)^{(j\epsilon -2)\rho m}\cdot (2-i\epsilon)^{(2-i\epsilon)\rho m}\ ,
\end{align*}
and we write 
    \begin{align*}
          \log_2 &\left( \frac{p_{v,u}}{\prod_{k=i}^{j-1}1/\gamma_k}\right)/m -O(\log (m)/m)\\
          &\le -2\rho (2-x)+(1-\rho x)h\left(\frac{\rho(2-x)}{1-\rho x}\right) + \rho(y-2)\log_2(y-2) + \rho (2-x)\log_2(2-x)\\
          &=(2-x)\rho (\log_2(2-x)-2)+(1-\rho x)h\left(\frac{\rho(2-x)}{1-\rho x}\right) + \rho(y-2)\log_2(y-2):=g(x,y)\ .
    \end{align*}
We compute
\begin{equation*}
    \frac{\partial g}{\partial y} = (\rho/\log(2))\cdot (1+\log(y-2)) 
\end{equation*}
hence 
\begin{equation*}
    \frac{\partial g}{\partial y} \ge 0 \iff y\ge 2+1/e\ .
\end{equation*}
Recall that in this case, we have the constraints that $1\le  x\le 2$, $2\le y\le 3$, and $y\ge 4-x$. Let us assume that $y\le 2+\delta_2$ for some small $\delta_2<1/e$ (this is wlog to prove our lemma, which needs to hold only for a small $\xi>0$). Hence, we need to evaluate the function $g(x,y)$ for $y=4-x$ to obtain an upperbound.

We have  
\begin{align*}
    f_2(x):=g(x,4-x) &= (2-x)(2\rho) (\log_2(2-x)-1)+(1-\rho x)h\left(\frac{\rho(2-x)}{1-\rho x}\right) 
\end{align*}
One can compute $\lim_{x\rightarrow 2^-}f_2(x)=0$ and
\begin{equation*}
    f_2'(x)=\frac{\rho}{\log (2)}\cdot \left[\log(4)-2+ \log\left(\frac{\rho}{(1-\rho x )(2-x)}\right)\right]\ 
\end{equation*}
which is non-negative in some small interval $[\delta_2,2)$ (we even have $\lim_{x\rightarrow 2^-}f_2'(x)=\infty$). By continuity, there exists a small $\delta_2>0$ such that $\max_{2-\delta_2<x\le 2}f_2(x)=0$.

\paragraph{Wrapping-up Case 3.} If we select $\xi=\min(\delta_1,\delta_2)$, we obtain that our upper-bound on the number of paths between $v$ and $u$ holds for any $v\in L_{i\le 2/\epsilon}$, any $u\in L_{j\ge 2/\epsilon}$ such that $|i-j|\le \xi \ell/3$. This concludes the proof of Lemma \ref{lem:helper_lemma_path_hierarchy}.

\bibliographystyle{alpha}
\bibliography{references}

\end{document}